%% file: qdp-concur.tex
\documentclass[a4paper,UKenglish]{lipics}
\usepackage{microtype}%if unwanted, comment out or use option "draft"

\usepackage{amstext}    % defines the \text command, needed here
\usepackage{array}
\usepackage{amsmath}
\usepackage{amssymb}
\usepackage{color}
\usepackage{graphicx}
\usepackage{epstopdf}
\usepackage{times}
\usepackage[ruled]{algorithm2e}
\usepackage{multirow}
\usepackage{tikz}
\usetikzlibrary{quantikz}

\input{mymacro}
\usepackage{xcolor}
\title{Verification of Distributed Quantum Programs
	\footnote{This work is partially supported by the National Key R\&D Program of China (Grant No: 2018YFA0306 701) and the Australian Research Council (Grant Nos: DP180100691, DP210102449).}
}
\author[1]{Yuan Feng}
\author[1]{Sanjiang Li}
\author[1,2,3]{Mingsheng Ying}
\affil[1]{Centre for Quantum Software and Information, University of Technology Sydney, Australia}
\affil[2]{Institute of Software, Chinese Academy of Sciences, Beijing, China}            

\affil[3]{Department of Computer Science, Tsinghua University, Beijing, China}   

\authorrunning{Y. Feng, S. Li, and M. Ying} %mandatory. First: Use abbreviated first/middle names. Second (only in severe cases): Use first author plus 'et. al.'

\Copyright{Yuan Feng, Sanjiang Li, and Mingsheng Ying}%mandatory, please use full first names. LIPIcs license is "CC-BY";  http://creativecommons.org/licenses/by/3.0/

\begin{document}
	
	\maketitle
	
	\begin{abstract}
		Distributed quantum systems and especially the Quantum Internet have the ever-increasing potential to fully demonstrate the power of quantum computation. This is particularly true given that developing a general-purpose quantum computer is much more difficult than connecting many small quantum devices. One major  challenge of implementing distributed quantum systems is programming them and verifying  their correctness. In this paper, we propose a CSP-like distributed programming language to facilitate the specification and verification of such systems. After presenting its operational and denotational semantics, we develop a Hoare-style logic for distributed quantum programs and establish its soundness and (relative) completeness with respect to both partial and total correctness. The effectiveness of the logic is demonstrated by its applications in verification of quantum teleportation and local implementation of non-local CNOT gates, two important algorithms widely used in distributed quantum systems.
	\end{abstract}
	
	\section{Introduction}\label{Intro}
	
	Quantum computers exploit quantum phenomena such as superposition and entanglement to perform computation. The past five years have seen exciting progresses in building small-scale quantum processors and the two state-of-the-arts, Google's Sycamore and IBM Q Rochester, both have 53 qubits. While these small %Noisy Intermediate-Scale Quantum (NISQ) 
	quantum devices already demonstrate certain advantages over classical supercomputers, large scale general-purpose quantum computers are still far from reach. 
	%Indeed, solving problems such as integer factorisation may require  millions of physical qubits, while the near-term %NISQ 
	%quantum devices have at most several thousands qubits. 
	
	The Quantum Internet has been  proposed as a  key strategy to provide large-scale quantum computing \cite{Kimble08,Wehner18internet,KozlowskiW19,Cuomo+20}. The idea is to connect many small quantum devices %(some may have only one qubit) 
	by using quantum communications and this network of quantum devices will then have the functionality of a (virtual) large-scale quantum computer. %Apparently, the number of qubits in this virtual quantum computer scales linearly with the number of quantum devices it connects. 
	%It is anticipated that a prototype of the Quantum Internet will be within reach over the next decade. Unlike NISQ devices, the research in the Quantum Internet is mainly promoted by governments and academia.
	On July 3, 2020, the Department of Energy of the United States proposed a 10-year roadmap for a national Quantum Internet under the \$1.2 billion \emph{National Quantum Initiative Act}. Several important steps have been experimented in the past two years. In February 2020, scientists from Argonne and the University of Chicago successfully entangled photons across a 52-mile underground network of optical fibre. %\footnote{https://news.uchicago.edu/story/argonne-uchicago-scientists-take-important-step-developing-national-quantum-internet} 
	In April 2021, a team of researchers from QuTech in the Netherlands reported realisation of the first entanglement-based quantum network (connecting three quantum processors) \cite{Pompili-3node}.
	
	%Quantum Internet ... The report lays out crucial research objectives, including building and then integrating quantum networking devices, perpetuating and routing quantum information, and correcting errors. Then, to put the nationwide network into place, there are four key milestones: verify secure quantum protocols over existing fiber networks, send entangled information across campuses or cities, expand the networks between cities, and finally expand between states, using quantum “repeaters” to amplify signals. ... recent progress in quantum internet ... Google's Sycamore, quantum supremacy weakly demonstrated, bottleneck, general-purpose quantum computer ... NISQ devices ...  more scalable if we connect small quantum devices ... secure and potential applications ... correctness verification challenge ...   data and gate telepotation
	
	%The Quantum Internet is a typical example of distributed quantum systems. Like its classical counterparts, distributed quantum systems will play an important role in human activities ranging from scientific discovery to earthquake forecasting, especially when security is a major concern. On the other hand,

	As pointed out in \cite {KozlowskiW19}, software-defined networking (SDN) technology is particularly important for quantum networks, because under current technical conditions, quantum memories have a very short lifespan.	
	On the other hand, programming quantum networks is much harder and more error-prone than programming classical ones   
	due to the possible existence of entanglement between different systems and non-commutativity of quantum observables and operations.
	
	Inspired by Apt's work~\cite{Apt86} on distributed programming based upon Hoare's CSP (Concurrent Sequential Processes)~\cite{hoare1978communicating}, 
	we define in this paper a programming language for distributed quantum systems. Recall that a distributed system consists of a number of spatially separated processes that work independently using their private storage, but communicate by explicit message passing. Our language supports both classical and quantum operations of individual processes. However, to make the presentation simpler, we only consider classical communication between different processes. Note that this is not a serious limitation, as generic quantum communication can be achieved by using the teleportation protocol~\cite{bennett1993teleporting} provided that entanglement is pre-shared between relevant parties.
	Furthermore, communication is achieved in a handshaking (or rendezvous) way; that is, the sender can deliver a message only when the receiver is ready to accept it at the same moment. We leave the asynchronous communication of quantum states as future work.
	%
	%This language supports both classical and quantum variables, which is a characteristic of the classical-quantum while language introduced in \cite{feng2020quantum}. The semantics of distributed quantum programs is then defined in a similar way as in \cite{feng2020quantum}. 
	%\delete{Due to the interleaving nature of local actions of individual processes and communication between disjoint pairs of processes, it is likely that the whole program may exhibit nondeterminism even if each individual process is deterministic. We prove that these different computations actually compute the same classical-quantum state (cq-state), thus clear the obstacle for defining the denotational semantics of distributed quantum programs.} 
	Based on the notion of classical-quantum assertions defined in \cite{feng2020quantum}, we propose Hoare-style logic systems for both partial and total correctness of distributed quantum programs, and prove their soundness and (relative) completeness. The effectiveness of these logic systems are demonstrated through the verification of %quantum data and gate teleportation, viz., \emph{teledata} and \emph{telegate}, 
	{quantum teleportation and local implementation of non-local CNOT gates,}
	two important algorithms widely used in distributed quantum systems. 
	%The investigation is closely related to \cite{feng2020quantum}, which prepares us with foundational notions like cq-states and cq-assertions required in this work. 
	It is worth noting that since the language we consider includes probabilistic assignments, this paper actually provides a sound and relatively complete Hoare logic for distributed probabilistic programs as a by-product.

	\textbf{Technical Contributions}:	While the semantics and proof systems in this paper are defined in a way similar to that of \cite{feng2020quantum}, the extension from sequential quantum programs to distributed quantum programs is challenging. 
	
	Firstly, the operational semantics of quantum measurements and probabilistic assignments in~\cite{feng2020quantum} are given in a `nondeterministic' way, with the probabilities of different branches being encoded in the quantum part of the configurations. This follows a tradition originated in~\cite{selinger2004towards} and adopted in~\cite{ying2012floyd,ying2016foundations} that  simplifies both notationally and conceptually the semantics of (deterministic) quantum languages,  especially the description of non-termination. However, distributed quantum programs investigated in this paper exhibit real nondeterminism (in the transition systems for operational semantics) due to the possible interleaving of local actions and communication of different sequential processes. To distinguish these two types of nondeterminism, we model quantum measurements and probabilistic assignments in a (standard) probabilistic way. Accordingly, the transition relation between configurations has to be lifted to probability distributions of configurations.

	Secondly, %due to the interleaving nature of local actions of individual processes and communication between disjoint pairs of processes, it is likely 
	despite that the entire distributed program may exhibit nondeterminism even if each individual process is deterministic, we show that different computations from a given configuration actually obtain the same classical-quantum state, thanks to the disjointness of the (classical changeable and quantum) variables accessible by different processes. This result clears the obstacle in defining the denotational semantics of distributed quantum programs and ensures that a distributed program can be sequentialised into a deterministic one without affecting its semantics. 
	
	Thirdly, the proof systems presented in~\cite{feng2020quantum} are designed for sequential quantum programs. New techniques are developed in this paper in extending them to distributed programs and proving their soundness and relative completeness.
	
	%Extending them to distributed programs, especially the proof for their soundness and relative completeness, requires the introduction of new techniques.
	
	\textbf{Organisation of the paper}: In the rest of this section, we briefly discuss some related works and present quantum teleportation as a motivating example. The remainder of this paper is organised as follows.  In Sec.~\ref{sec:syntax}, we present the three layers of the syntax of the distributed quantum programming language, which is followed by its operational and denotational semantics %of quantum programs with classical variables 
	in Sec.~\ref{sec:semantics}. In particular, we prove that distributed quantum programs are semantically deterministic in the sense that different computations from a given configuration always give the same classical-quantum state. We then show in Sec.~\ref{sec:sequentialisation} how a distributed quantum program can be sequentialised without affecting its semantics. Based on the notion of classical-quantum assertion, we present a Hoare-style logic in Sec.~\ref{sec:verification} for distributed quantum programs and establish its soundness and (relative) completeness for both partial and total correctness. The last section concludes this paper with an outline of future works. Due to space limitation, we omit all proofs as well as the verification of  %quantum data and gate teleportation. 
	{quantum teleportation and local implementation of non-local CNOT gates}.
	Interested readers may find these details in the appendix.

	\subsection{Related Works}
	
	The {following three} lines of previous works are closely related to this paper.
	
	\textbf{Quantum Process Algebras}: 
	% 	Process algebra is the mainstream approach to formally model and reason about distributed quantum systems. Since 2004, several quantum process algebras like  QPAlg \cite{JorrandL04}, CQP \cite{GayN05}, and qCCS \cite{feng2007probabilistic,ying2009algebra,FengDY12} have been introduced and adopted in verification of popular quantum communication protocols such as Teleportation~\cite{bennett1993teleporting} and Superdense Coding~\cite{bennett1992communication}. While quantum process algebras are suited to high-level formal specification of distributed quantum computing, quantum Hoare Logic is in spirit more close to low-level circuit implementation and often easier to apply. 
	% 	In this paper, we adopt a Hoare logic approach to verify the correctness of distributed quantum systems. 
	Process algebra is the mainstream approach to formally model and reason about quantum communication systems. Since 2004, several quantum process algebras such as QPAlg \cite{JorrandL04}, CQP \cite{GayN05}, and qCCS \cite{feng2007probabilistic,ying2009algebra,FengDY12} have been introduced and adopted in verification of popular quantum communication protocols such as teleportation~\cite{bennett1993teleporting} and superdense coding~\cite{bennett1992communication}. 
	Following~\cite{Apt86} (also see~\cite{apt2010verification}, Chapter 11), we choose to use (a subset of) a quantum extension of process algebra CSP as our language for programming distributed quantum systems, but use a Hoare-style logic to reason about their correctness. 
	%However, process algebras abstract away the details of local computation and focus on communication between concurrent processes, thus they are not suitable for our purpose, which is the verification of distributed quantum programs. 
	
	\textbf{Quantum Hoare Logic}: 
	%Why previous quantum Hoare logics are not convenient for quantum networks?
	Hoare logic provides a syntax-oriented proof system to reason about program correctness~\cite{hoare1969axiomatic}. In recent years, Hoare-style logics for  quantum programs have been developed in~\cite{chadha2006reasoning,feng2007proof,Kakutani:2009,ying2012floyd,unruh2019quantum,feng2020quantum}. However, these logic systems are designed for the verification of sequential quantum programs, thus are not suitable for the distributed ones considered in the current paper.  
	Nevertheless, our definition of semantics of distributed quantum programs is based on the key notions such as classical-quantum states and assertions introduced in~\cite{feng2020quantum}. 
	
	%The proof rules for sequential quantum programs (not the general distributed ones) in this paper, except for the ones for alternative and repetitive commands, are also borrowed from~\cite{feng2020quantum}.
	
	% 	\textbf{Quantum Hoare Logic}: 
	% 	%Why previous quantum Hoare logics are not convenient for quantum networks?
	% 	\red{As classical communications are necessary in quantum networks, we need a quantum Hoare logic that supports classical variables. However, most existing quantum Hoare logics that support classical variables are incomplete. Recently, this important missing link was completed in \cite{feng2020quantum}, which provides a sound and complete quantum Hoare logic for a simple while language supporting both classical and
	% 		quantum variables. The current investigation is closely related to \cite{feng2020quantum}. In particular, the key notions like cq-states and cq-assertions are introduced there.  
	% 	}
	
	\textbf{Programming with Quantum Communication}:  The authors of~\cite{tafliovich2009programming} presented some interesting ideas of specifying and analysing quantum communication in a predicative programming language. However, the key technique for verification of quantum communication protocols developed in~\cite{tafliovich2009programming} (and in predicative programming~\cite{hehner2012practical} in general) is refinement, while we use a Hoare-style logic here.
	
	\subsection{Motivating Example --- Quantum Teleportation}\label{sec:telep}
	Quantum teleportation was proposed by Bennett et al. \cite{bennett1993teleporting} for transmitting \textit{quantum information} (e.g. the exact state of an atom or photon) via only \textit{classical communication} but with the help of previously shared \textit{quantum entanglement} between the sender and the receiver. It is one of the most surprising examples where entanglement helps to accomplish a certain task that is impossible in the classical world. A large number of quantum communication protocols %including BB84 \cite{bennett1984quantum} and Superdense Coding \cite{bennett1992communication}
	such as quantum gate teleportation~\cite{gottesman1999demonstrating}, port-based 
	teleportation~\cite{ishizaka2008asymptotic}, quantum repeaters~\cite{briegel1998quantum}, and measurement based quantum computing~\cite{raussendorf2001one}	have been designed based on it, and some of them have been experimentally implemented~\cite{pirandola2015advances}.
	%(see, e.g., \cite{Williams17-implm-superdense}). 
	
	%In this section, for convenience of the reader, we review  the protocol of quantum teleportation and its various generalisations . 
	
	%\subsection{Teleporting a Qubit}
	
	Let us consider the simplest case of teleporting a qubit.  Assume that Alice and Bob live far apart and there is only a classical communication channel between them. But Alice wants to send quantum information, say a state $|\psi\rangle\define \alpha_0|0\> + \alpha_1|1\>$ of qubit $q$, to Bob. How can she do it? This seems a task impossible for her to accomplish because it may take infinite amount of classical information to describe the complex amplitudes $\alpha_0$ and $\alpha_1$. However, if Alice and Bob share entanglement; more precisely, if they possess qubits $q_1$ and $q_2$ respectively and these two qubits are in the Bell state $|\beta\rangle\define \frac{1}{\sqrt{2}}(|00\> + |11\>)$ (also called EPR pair), then they can accomplish the task using the following protocol, called \textit{teleportation}:
	\begin{enumerate}\item Alice interacts qubit $q$ in state $|\psi\rangle$ and her half $q_1$ of the shared EPR pair $|\beta\rangle$ by performing first the controlled NOT (CNOT for short) on $q,q_1$ and then the Hadamard gate $H$ on $q$, where: 
		\begin{itemize}
			\item the CNOT acts as follows: if the control qubit $q$ is in $|0\rangle$ then the target qubit $q_1$ is left unchanged, and if $q$ is in $|1\rangle$ then $q_1$ is flipped between $|0\>$ and $|1\>$;
			\item the $H$ gate turns basis states $|0\rangle$ and $|1\rangle$ to their equal superposition $|+\>$ and $|-\>$, where $|\pm\rangle=\frac{1}{\sqrt{2}}(|0\rangle\pm |1\rangle)$, respectively.  
		\end{itemize}
		\item Alice measures her qubits $q,q_1$ (in the standard basis), and sends the obtained results -- classical bits $z,x$ through the classical channel to Bob.
		\item On his half $q_2$ of the EPR pair, Bob performs operation $X$ whenever the received classical information $x=1$, and then $Z$ whenever $z=1$. Here $X$ and $Z$ are Pauli operations with $X|i\> = |1-i\>$ and $Z|i\> = (-1)^i |i\>$ for $i=0,1$.
	\end{enumerate}           
	Quantum teleportation can be visualised as the quantum circuit in Figure \ref{fig-teleportation}. 
	\begin{figure}[t]\centering
		\tikzset{
			my label/.append style={above right,xshift=0.3cm}
		}
		\begin{quantikz}[row sep=0.3cm,column sep=1cm]
			\lstick{$|\psi\>$} &\ctrl{1} & \gate{H} &  \meter{$z$} & \cw & \cwbend{3}\\	  
			\lstick[3]{$|\beta\>$} &\targ{}  &   \qw &  \meter{$x$} & \cwbend{2}\\
			&\wave&&&&&&\\
			&\qw  &  \qw &  \qw & \gate{X} &\gate{Z} &\qw 
		\end{quantikz}
		\caption{Quantum Teleportation. The wires from top to bottom represent qubits $q$, $q_1$, and $q_2$ respectively. Furthermore, $q$ and $q_1$ belong to Alice while $q_2$ belongs to Bob.}\label{fig-teleportation}
	\end{figure}
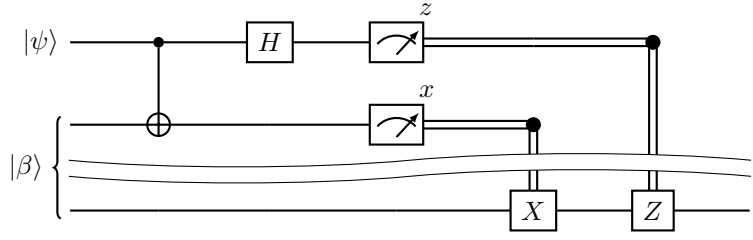
	What surprises us is that at the end Bob's qubit $q_2$ is in state $|\psi\rangle$. In other words, Alice sends the quantum information $|\psi\rangle$ to Bob only by classical communication of two bits in step (2), even without knowing the amplitudes $\alpha_0$ and $\alpha_1$ of $|\psi\rangle$. Of course, this is achieved by consuming some entanglement (At the end of the protocol, qubits $q_1$ and $q_2$ are no longer entangled).  
	%
	%	
	%	\begin{itemize}\item A \textbf{distributed system} consists of a number of physically distributed components that work independently using their private storage, but communicate by explicit message passing. 
	%		
	%		\item A \textbf{distributed program} is an abstract description of a distributed system. It consists of a collection of processes:\begin{itemize} \item Each process can access a set of variables which are disjoint from the variables that can be changed by other processes;
	%			\item Two modes of message passing:
	%			\begin{itemize}\item \textit{Synchronous communication (handshaking or rendezvous)}: The sender of a message can deliver it only when the receiver is ready to accept it at the same moment;
	%				\item \textit{Asynchronous communication}: The sender can always deliver its message, and an implicit buffer is assumed where messages can be kept until the receiver collect them.
	%			\end{itemize}
	%	\end{itemize}\end{itemize}
	%	
	%	\textbf{We focus on LOCC, Teleportation-Based Quantum Communication Systems}. Including quantum communication makes the language more complicated (we have to trace the ownership of each qubit so that the no-cloning property of quantum information will not be violated).

	\section{A Language for Programming Distributed Quantum Systems}
	\label{sec:syntax}
	
	We propose a programming language to describe distributed quantum systems. The syntax has three layers, introduced in the following three subsections respectively. 
	
	\subsection{Sequential quantum programs}\label{sec:seq}
	
	For the first layer, we extend the classical-quantum while language defined in~\cite{feng2020quantum} with alternative and repetitive commands~\cite{dijkstra1976discipline}. We assume two basic types for classical variables: $\tybool$ with the corresponding domain $D_{\tybool} \define \{\true, \false\}$ and $\tyint$ with $D_{\tyint} \define\Z$. For each integer $d\geq 1$, we assume a basic quantum type $\tyqudit$ with domain $\h_{\tyqudit}$, which is a $d$-dimensional Hilbert space with an orthonormal basis $\{|0\>, \ldots, |d-1\>\}$. In particular, we denote the quantum type for $d=2$ as $\tyqubit$. Let $\CVar$, ranged over by $x,y,\cdots$, and $\QVar$, ranged over by $q, r, \cdots$, be countably infinite sets of classical and quantum variables, respectively. Denote by $\mathit{type}(v)$ the type of a (classical or quantum) variable $v$. For any
	finite subset $V$ of $\QVar$, let
	$\h_V \define \bigotimes_{q\in V} \h_{q},
	$
	where $\h_{q} \define \h_{\type(q)}$. In this paper, when we refer to a subset of $\QVar$, it is always assumed to be finite.
	
	With the above notions, a sequential quantum program is defined by the following syntactic rules:
	\begin{align*} S::= &\ \sskip\ |\ \abort\ |\ x:= e\ |\ x\rassign g\ | \ x:= \mymeas\ \m[\bar{q}]\ | \ q:=0\ |\
		\bar{q}\apply U\ |\ S_0;S_1\ |\\ &\ \alterex\ |\ \repex
	\end{align*}
	where $S$ and $S_i$ are sequential quantum programs, $x$ a classical variable in $\CVar$, $e$ a classical expression with the same type as $x$, $g$ a discrete probability distribution over $D_{\mathit{type}(x)}$, $B_i$ a $\tybool$-type expression, $q$ a quantum variable and $\bar{q} \define q_1, \ldots, q_n$ a (ordered) tuple of distinct quantum variables in $\QVar$, $\m$ a measurement and $U$ a unitary operator on
	$d_{\bar{q}}$-dimensional Hilbert space with $$d_{\bar{q}}\define \dim(\h_{\bar{q}}) = \prod_{i=1}^{n} \dim(\h_{q_i}).$$
	Sometimes we also use $\bar{q}$ to denote the (unordered) set $\{q_1,q_2,\dots,q_n\}$. Let $|\bar{q}|\define n$ be the size of $\bar{q}$. 
	We write $x:= \mymeas\ \bar{q}$ for  $x:= \mymeas\ \m_{\mathit{com}}[\bar{q}]$ 
	where $\m_{\mathit{com}} \define \{P_k \define |k\>\<k| : 0\leq k<d_{\bar{q}}\}$ is the projective measurement according to the computational basis of $\h_{\bar{q}}$. We always write $|k\>$ for the product state $|k_1\>\cdots|k_n\>$, where $k = \sum_{i=1}^{n}k_i  d_{q_{i+1}}\ldots d_{q_n}$.

	The alternative and repetitive commands above are sometimes abbreviated as 
	\[
	\altercom\hspace{2em} \mbox{  and  }\hspace{2em} \repcom
	\]
	respectively. For simplicity, we only consider deterministic sequential quantum programs in this paper. To this end, we assume that the $B_i$'s are mutually exclusive; that is, for each $i$, $B_i\ra \bigwedge_{j\neq i} \neg B_j$ is a tautology. However, we do not require $\bigvee_{i=1}^n B_i \leftrightarrow \true$.
	Under this assumption, a guarded command $B_i\ra S_i$ in $\altercom$ will be chosen to execute once its guard $B_i$ evaluates to $\true$.  If all guards evaluate to $\false$, the alternative command will lead to a (classical) \emph{failure} state, which is a feature introduced in~\cite{dijkstra1976discipline} but does not exist in the while language of~\cite{feng2020quantum}. The selection of guarded commands in $\repcom$ follows a similar way, with the only difference that after termination of a selected $S_i$ the whole command is repeated. Moreover, in contrast with the alternative command, the repetitive command properly terminates if all the guards evaluate to false.

	\subsection{Sequential quantum process}
	To describe the second syntactic layer for distributed quantum programs, we adopt a subset of Hoare's CSP (Communicating Sequential Processes)~\cite{hoare1978communicating,brookes1984theory}, following the approach in~\cite{apt2010verification}.  Let $\Chan$
	be a set of (classical) channel names, ranged over by $c,d,\dots$.
	An \emph{input command} is of the form $c?x$, while an \emph{output command} is of the form $c!e$,
	where $c\in \Chan$ is a communication channel, $x\in \CVar$ a classical variable, and $e$ an expression.
	Intuitively,  $c?x$ expresses the request to receive a classical value along channel $c$. Upon reception this value is assigned to variable $x$. In contrast, $c!e$ expresses the request to send the value of expression $e$ along channel $c$. 
	A \emph{generalised guard} is of the form $g\define B;\alpha$ where $B$ is a Boolean expression, and $\alpha$ an input or output command. In particular, if $B\equiv\mathbf{true}$, then we denote $g$ simply as $\alpha$.

	Let $\alpha_1$ and $\alpha_2$ be two input/output (i/o) commands. They are said to \emph{match} if they refer to the same channel, one of them is an input, and the other one output with the same type. Given two matched i/o commands  $\alpha_1\define c?x$ and $\alpha_2\define c!e$, the \emph{communication effect} of $\alpha_1$ and $\alpha_2$ is defined to be the program statement $x:=e$; that is, $$\mathit{Effect}(\alpha_1,\alpha_2)= \mathit{Effect}(\alpha_2,\alpha_1)\define x:=e.$$
	
	%Similarly, quantum input and output commands have the form of 
	%$\qc{c}?q$ and $\qc{c}!q$ respectively. Note that due to the no-cloning of quantum information, quantum communication is realised through transmission of the physical system $q$.

	\begin{definition} A \emph{sequential quantum process} has the form:
		$$S::= S_0; \mathbf{do}\ \square_{j=1}^m B_j;\alpha_j\rightarrow S_j\ \mathbf{od}$$
		where $m\geq 0$, $S_0,S_1,\ldots,S_m$ are sequential quantum programs defined in the previous subsection. {Again, we assume that  $B_j$'s are mutually exclusive.} 
		We call $S_0$ the \emph{initialisation part}, and $\mathbf{do}\ \square_{j=1}^m B_j;\alpha_j\rightarrow S_j\ \mathbf{od}$ the \emph{main loop} of $S$. If $m=0$, then we let $S=S_0$. In this way, any sequential quantum program is a sequential process. 
		If $S_0 \equiv \sskip$, we drop $S_0$ from $S$ unless $m=0$.
	\end{definition}

	We have the following notations for sequential quantum process $S$.
	\begin{itemize}\item Denote by $\cVar(S)$ and $\qVar(S)$ the sets of classical and quantum variables appearing in $S$, respectively. Note that we do not distinguish between free and bound variables; that is, the classical variable appearing in an input command of $S$ is also included in $\cVar(S)$. Let $\Var(S) \define \cVar(S)\cup \qVar(S)$.
		\item Denote by $\change(S)$ the set of classical variables that appear on the left-hand side of an assignment or in an input command in $S$. Note that the only way to retrieve information from a quantum system is to measure it, which may change its state. Thus $\qVar(S)$ is also the set of changeable quantum variables in $S$. 
		\item Denote by $\Chan(S)$ the set of channel names appearing in $S$.
	\end{itemize}

	\subsection{Distributed quantum programs}
	
	Now we are ready to define the syntax for distributed quantum programs.
	\begin{definition}\label{def:distprog}A \emph{distributed quantum program} is a parallel composition $S::=S_1\|\cdots\|S_n$ where $n\geq 1$ and $S_1,\ldots,S_n$ are sequential quantum processes defined in the above subsection which satisfy
		\begin{itemize}\item Pairwise disjointness: for all $1\leq i\neq j\leq n$, 
			$\Var(S_i)\cap (\change(S_j) \cup \qVar(S_j))=\emptyset;$
			\item Point-to-point connection: for all $1\leq i<j<k\leq n$,
			$\Chan(S_i)\cap\Chan(S_j)\cap\Chan(S_k)=\emptyset.$
		\end{itemize}
		Let $\cVar(S) \define \bigcup_{i=1}^n \cVar(S_i)$, and $\change(S)$, $\qv(S)$, and $\Var(S)$ be similarly defined.
	\end{definition}
	
	Essentially, the first clause requires that (1) classical variables in any process cannot be changed by other processes; (2) quantum variables in any process do not appear in other processes. The second clause in Definition~\ref{def:distprog} implies that each communication channel is shared by at most two processes. This constraint, together with the assumption that sequential processes are deterministic, means that at any moment, each process is only able to communicate with at most one other process. Note also that we disallow nested parallelism in distributed programs. Finally, any sequential quantum process is a distributed quantum program with $n=1$.
	
	The constraints in Definition~\ref{def:distprog} look very strict at the first glance. However, using similar approaches presented in~\cite{apt1987two,zoebel1988normalform}, more general distributed quantum systems can be transformed into this special form by introducing control variables (say, $stage_A$ and $stage_B$ in the following example).
	\begin{example}[Quantum Teleportation as a Distributed Program]
		The quantum teleportation protocol presented in Sec.~\ref{sec:telep} can be written as a distributed program
		$\textit{Teleport}\define \mathit{Alice}\ \|\ \mathit{Bob}$
		where $\mathit{Alice}\define$
		\begin{align*}
			& q,q_1 \apply \textit{CNOT};\ q\apply H;\  z_A := \mymeas\ q;\  x_A := \mymeas\ q_1;\ 
			stage_A := 0;\\
			& \mathbf{do}\ stage_A = 0; c!x_A \ra stage_A := 1
			\  \square\  stage_A = 1; d!z_A \ra stage_A := 2\ 
			\mathbf{od}
		\end{align*}
		and $\textit{Bob}\define$
		\begin{align*}
			& stage_B := 0;\\
			& \mathbf{do}\ stage_B = 0; c?x_B \ra stage_B := 1; \ \measstm{x_B=1}{q_2\apply X}{\sskip} \\
			&\ \  \square\  stage_B = 1; d?z_B \ra stage_B := 2; \ \measstm{z_B=1}{q_2\apply Z}{\sskip} \\
			& \mathbf{od}
		\end{align*}
		
	\end{example}
	
	\section{Operational and Denotational Semantics}
	\label{sec:semantics}
	
	We recall some basic notions from~\cite{feng2020quantum} to define the semantics of distributed quantum programs.
	\subsection{Classical-quantum states}

	Let $\Sigma \define \CVar \rightarrow D$ be the (uncountably infinite) set of \emph{classical states}, where $D\define D_{\tybool}\cup D_{\tyint}$.  We further require that states in $\cstates$ respect the types of classical variables; that is, $\cstate(x) \in D_{\mathit{type}(x)}$ for all $\cstate\in \cstates$ and $x\in \CVar$. For $V\subseteq \QVar$, let $\dhv$ be the set of partial density operators on $\h_V$; that is, positive linear operators with the trace being less than or equal to 1. Furthermore, let $\z_{\h_V}\in \dhv$ be the zero operator on $\h_{V}$.
	
	\begin{definition}\label{def:cqstate}
		Given $V\subseteq \QVar$, a \emph{classical-quantum state} (cq-state for short) $\qstate$ on $V$ is a function in $\cstates\rightarrow\dhv$ such that  
		\begin{enumerate}
			\item the support of $\qstate$, denoted $\supp{\qstate}$, is countable. That is, $\qstate(\cstate) \neq \z_{\h_V}$ for at most countably infinite many $\cstate\in \cstates$;
			\item $	\tr(\qstate) \define \sum_{\cstate\in \supp{\qstate}}\tr[\qstate(\cstate)] \leq 1$.
		\end{enumerate}
	\end{definition}
	
	Denote by $\qVar(\qstate)$ the set $V$ of quantum variables in $\qstate$ defined in Definition~\ref{def:cqstate}.
	%	Recall that a (probability) distribution $\mu$ over $\cstates$ is a function in $\cstates \ra [0,1]$ with countable support such that $\sum_{\cstate\in \supp{\mu}} \mu(\cstate) = 1$. Such a distribution can be regarded as a cq-state on $V = \emptyset$, as in this case $\dhv = [0,1]$. 
	Sometimes it is convenient to denote a cq-state $\qstate$ by the explicit form $\bigoplus_{i\in I}\<\cstate_i, \rho_i\>$ where
	$\supp{\qstate}= \{\cstate_i : i\in I\}$ 
	and $\qstate(\cstate_i)=\rho_i$ for each $i\in I$. 
	When $\qstate$ is a simple function such that $\supp \qstate=\{\cstate\}$ for
	some $\cstate$ and $\qstate(\cstate)=\rho$, we denote $\qstate$ simply by $\<\cstate, \rho\>$. 
	Let  $\{\qstate_i : i\in I\}$ be a countable set of cq-states over $V$ such that  for any $\cstate$,
	$\sum_{i\in I} \qstate_i(\cstate) = \rho_\cstate$ for some $ \rho_\cstate\in \dhv$ and $\sum_{i\in I} \tr(\qstate_i) \leq 1$. Then the summation of them, denoted $\sum_{i\in I} \qstate_i$, is a cq-state $\qstate$ over $V$
	such that for any $\cstate\in \cstates$, $\qstate(\cstate)= \rho_\cstate$. Obviously, $\supp{\qstate} = \bigcup_{i\in I}\supp{\qstate_i}$. 
	It is worth noting the difference between $\sum_{i\in I} \<\cstate_i, \rho_i\>$, the summation of some (simple) cq-states, and  $\bigoplus_{i\in I} \<\cstate_i, \rho_i\>$, the explicit form of a single one: in the latter $\cstate_i$'s must be distinct while in the former they may not.

	Let $\qstatesh{V}$ be the set of all cq-states over $V$, and $\s$ the set of all cq-states; that is, 
	$\s \define \bigcup_{V\subseteq \QVar} \qstatesh{V}.$ 
	We extend the L\"{o}wner order $\le_V$ for $\d(\h_V)$ pointwisely to $\s$ by letting $\qstate\le \qstate'$ iff $qv(\qstate) = qv(\qstate')$ and for all $\cstate\in \Sigma$, $\qstate(\cstate) \le_{qv(\qstate)} \qstate'(\cstate)$. 
	Then $\qstatesh{V}$ is a pointed $\omega$-CPO under $\le$, with the least element being the constant $\z_{\h_V}$ function, denoted $\emptydis_V$. Furthermore, $\s$ as a whole is an $\omega$-CPO under $\le$.
	When $\qstate \le \qstate'$, there exists a unique $\qstate''\in \qstatesh{qv(\qstate)}$, denoted $\qstate' - \qstate$, such that $\qstate'' + \qstate = \qstate'$. 
	For any real numbers $\lambda_i$, $i\in I$, if both $\qstate_+ \define \sum_{\lambda_i >0}\lambda_i  \qstate_i$ and $\qstate_- \define \sum_{\lambda_i <0}(-\lambda_i)  \qstate_i$ are well-defined and $\qstate_- \le  \qstate_+$ ,
	then the linear-sum $\sum_{i\in I}\lambda_i  \qstate_i$ is defined to be $\qstate_+ - \qstate_-$. In the rest of this paper, whenever we write $\sum_{i\in I}\lambda_i  \qstate_i$ we always assume that it is well-defined.
	Finally, let $\e$ be a completely positive and trace-nonincreasing super-operator from $\l(\h_V)$ to $\l(\h_W)$. We extend it to $\qstatesh{V}$ in a pointwise way: $\e(\qstate)(\cstate) = \e(\qstate(\cstate))$ for all $\cstate$.
	\subsection{Operational Semantics}

	\newsavebox{\tablebox}
	\begin{table}[t]
		\begin{lrbox}{\tablebox}
			\centering
			\begin{tabular}{c}
				{\renewcommand{\arraystretch}{2.5}
					\begin{tabular}{ll}
						$\<\sskip, \cstate, \rho\> \ra \<E, \cstate, \rho\>$ & $\<x:= e, \cstate, \rho\> \ra \<E, \cstate[\subs{x}{\cstate(e)}], \rho\>$ \\
						$\<q:=0, \cstate, \rho\> \ra \<E, \cstate, \sum_{i=0}^{d_{q}-1}\qzi \rho\qiz\>$ & $\<\bar{q}\apply U, \cstate, \rho\> \ra \<E, \cstate, U_{\bar{q}} \rho U_{\bar{q}}^\dag\>$\\
						$\displaystyle\frac{}{\<x\rassign g, \cstate, \rho\> \ra \sum_{d\in D_{\mathit{type}(x)}}g(d)\cdot \<E, \cstate[\subs{x}{d}], \rho\>}$ &
						$\displaystyle\frac{\m = \{M_i : i\in I\},\ \rho_i = M_i \rho M_i^\dag,\ p_i = \tr(\rho_i) }{\<x:= \mymeas\ \m[\bar{q}], \cstate, \rho\> \ra \sum_{p_i>0} p_i \cdot \<E, \cstate[\subs{x}{i}], \rho_i/p_i\>}$\\
						$\displaystyle\frac{\cstate \models B_i, 1\leq i\leq n}{\<\altercom, \cstate, \rho\> \ra  \<S_i, \cstate, \rho\>}$
						& $\displaystyle\frac{\cstate \models \bigwedge_{i=1}^n \neg B_i}{\<\altercom, \cstate, \rho\> \ra  \<E, \fail, \rho\>}$\\
						$\displaystyle\frac{\cstate \models B_i, 1\leq i\leq n}{\<\repcom, \cstate, \rho\> \ra  \<S_i; \repcom, \cstate, \rho\>}$&
						$\displaystyle\frac{\cstate \models \bigwedge_{i=1}^n \neg B_i}{\<\repcom, \cstate, \rho\> \ra  \<E, \cstate, \rho\>}$\\
						$\displaystyle\frac{\<S_0, \cstate, \rho\> \ra \sum_{i\in I}p_i\cdot \<S_i, \cstate_i, \rho_i\>}{\<S_0; S_1, \cstate, \rho\> \ra \sum_{i\in I}p_i\cdot \<S_i; S_1, \cstate_i, \rho_i\>}$\ where $E; S_1 \equiv S_1$ & $\displaystyle\frac{ \sigma\models\bigwedge_{j=1}^m\neg B_j}{\<\mathbf{do}\ \square_{j=1}^m B_j;\alpha_j\rightarrow S_j\ \mathbf{od},\sigma, \rho\>\rightarrow\<E,\sigma,\rho\>}$\\
				\end{tabular} }\\
				\\
				{\renewcommand{\arraystretch}{2.5}
					(Paral)\ \ $\displaystyle\frac{\<S_k, \cstate, \rho\> \ra \sum_{i\in I}p_i\cdot \<S_{k,i}, \cstate_i, \rho_i\>,\ 1\leq k\leq n}{\< S_1\|\ldots\|S_k\|\ldots\|S_n,\sigma,\rho\>\rightarrow \sum_{i\in I}p_i\cdot \< S_1\|\ldots\|S_{k,i}\|\ldots\|S_n, \sigma_i,\rho_i\>
					}$}\\
				\\
				{\renewcommand{\arraystretch}{1.5}
					(Comm)\ \ $\frac{\begin{array}{cc}S_k\equiv \mathbf{do}\ \square_{j=1}^m B_{k,j}; \alpha_{k,j}\rightarrow S_{k,j}\ \mathbf{od},\quad S_\ell\equiv\mathbf{do}\ \square_{j=1}^{m'} B_{\ell,j}; \alpha_{\ell,j}\rightarrow S_{\ell,j}\ \mathbf{od},\ 1\leq k<\ell\leq n\\
							\sigma\models B_{k,j_1}\wedge B_{\ell,j_2},\  \mbox{$\alpha_{k,j_1}$ and $\alpha_{\ell,j_2}$ match}, \mathit{Effect}(\alpha_{k,j_1}, \alpha_{\ell,j_2}) \equiv x := e,\ 1\leq j_1\leq m, 1\leq j_2\leq m'
					\end{array}}{\begin{array}{cc}\< S_1\|\ldots\|S_n,\sigma,\rho\>\rightarrow \< S_1'\|\ldots\|S_n', \sigma[\sigma(e)/x],\rho\>\\
							{\rm where}\ S_k'\define S_{k,j_1}; S_k,\ S_\ell'\define S_{\ell,j_2};S_\ell,\ {\rm and}\ S_i'\define S_i\ {\rm  for}\ i\neq k, \ell 
					\end{array}}$}
			\end{tabular} 
		\end{lrbox}
		\resizebox{\textwidth}{!}{\usebox{\tablebox}}\\
		\vspace{4mm}
		\caption{Operational semantics for distributed quantum programs, where $\cstate$ is a proper classical state; i.e., $\cstate\not\equiv \fail$. 
		}
		\label{tbl:opsemantics}
	\end{table}
	
	Let $\prog$ be the set of all distributed quantum programs. A \emph{configuration} is a triple $\<S, \cstate, \rho\>$ where
	$S\in \prog \cup \{E\}$ with $E$ being a special symbol to denote termination, $\cstate\in \Sigma\cup \{\fail\}$ with $\fail$ being another special symbol to denote the failure state, and $\rho\in \d(\h_V)$ for some $V$ subsuming $\qv(S)$ with $\tr(\rho) =1$. We always identify $E\|\ldots \| E$ with $E$. The operational semantics of programs in $\prog$ is defined as the smallest transition relation $\rightarrow$  given in Table~\ref{tbl:opsemantics}. 
	
	\begin{remark}
		The transition rules presented in Table~\ref{tbl:opsemantics}  for sequential quantum programs follows the same spirit as in~\cite{feng2020quantum}, except for the newly introduced alternative and repetitive commands whose semantics definitions are also standard~\cite{dijkstra1976discipline}. The rules (Paral) and (Comm) are similar to their analogy for classical non-probabilistic programs~\cite{apt2010verification}.
		
		It is worth noting that the transitions for quantum measurements and probabilistic assignments in~\cite{feng2020quantum} are given in a `non-deterministic' way, with the probabilities of different branches being encoded in the quantum part of the configurations (by allowing partial density operators instead of density operators in configurations). Note that it is only a matter of notational convenience to represent probabilistic choices with non-determinism. However, distributed quantum programs investigated in this paper exhibit real non-determinism due to the possible interleaving of local actions and communication of different sequential processes. To distinguish these two types of non-determinism, we decide to model quantum measurements and probabilistic assignments in a (standard) probabilistic way.	\qed
	\end{remark}

	The following lemma, which can be easily proved by inspecting the transition rules in Table~\ref{tbl:opsemantics}, shows that $\rightarrow$ is indeed a relation from configurations to probability distributions of configurations.

	\begin{lemma}
		Let $\<S, \cstate, \rho\>$ be a configuration and $\<S, \cstate, \rho\>\ra \sum_{i\in I}p_i\cdot \<S_{i}, \cstate_i, \rho_i\>$. Then $\sum_{i\in I} p_i =1$.
	\end{lemma}
	
	The next lemma extends the Change and Access lemma for classical programs by considering the effects of transitions on quantum states.
	\begin{lemma}[Change and Access]\label{lem:supstep}
		Let $\<S, \cstate, \rho\> \ra \mu$. Then there exist a set $\{S_i: i\in I\}$ of distributed programs with $v(S_i)\subseteq v(S)$ for $v\in \{\change, \qv, \cVar\}$, a set $\{f_i : i\in I\}$ of functions over $\cstates$, and a set $\{\e_i: i\in I\}$ of super-operators acting on $\h_{\qVar(S)}$ such that
		\begin{enumerate}
			\item for each $i$, $f_i$ does not change the value of variables outside $\change(S)$. That is, for all $\tau\in \Sigma$, $f_i(\tau)|_V =  \tau|_V$ where $V\define \CVar\backslash\change(S)$;		
			\item for each $i$, $f_i$ depends only on $\cVar(S)$. That is, $f_i(\cstate)|_{\cVar(S)} = f_i(\tau)|_{\cVar(S)}$ whenever  $\cstate|_{\cVar(S)} =  \tau|_{\cVar(S)}$;
			\item $\sum_{i\in I}\e_i$ is trace-nonincreasing; 
			\item $\mu = \sum_{i\in I, p_i>0}p_i\cdot \left\<S_{i}, f_i(\cstate), \e_i(\rho)/p_i\right\>$ where $p_i = \tr(\e_i(\rho))$;
			\item for any $\cstate'$ which agrees with $\cstate$ on $\cVar(S)$, i.e. $\cstate'|_{\cVar(S)} = \cstate|_{\cVar(S)}$, and $\rho'\in \dhv$ with $V\supseteq \qVar(S)$,  
			\begin{equation}\label{eq:tranform}
				\left\<S, \cstate', \rho'\right\>\ra \sum_{i\in I, p_i'>0}p_i'\cdot \left\<S_{i}, f_i(\cstate'), \e_i(\rho')/p_i'\right\>
			\end{equation} 
			where $p_i' = \tr(\e_i(\rho'))$. 
		\end{enumerate} 
	\end{lemma}

	A configuration is called a \emph{terminal} if it has no successor distributions. Because of the communication constraints, distributed programs can also end up with a \emph{deadlock} configuration, in which not all the processes terminate properly (become $E$), and none of them has led to a failure (the classical state becomes $\fail$). In other words, $\<S, \cstate, \rho\>$ is a terminal iff $S\equiv E$, $\cstate \equiv \fail$, or it is a deadlock. For a distribution $\mu= \sum_{i\in I} p_i \cdot  \<S_{i}, \cstate_i, \rho_i\>$ of configurations, we denote by 
	$$\qstate_\mu\define \sum_{i\in I, S_i \equiv E, \cstate_i\not\equiv \fail} \<\cstate_i, p_i \rho_i\>$$ 
	the cq-state obtained by restricting $\mu$ on the properly terminated configurations. Let $\Pr_\mu(E) \define \tr(\qstate_\mu)$ be the probability of $\mu$ having properly terminated.

	The transition relation $\ra$ defined above can be further extended to distributions of configurations by letting $\mu \ra \nu$ where $\mu = \sum_{i\in I} p_i \cdot c_i$ if 
	%(1) not all support configurations of $\mu$ are terminal, 
	(1) for each $i$, $c_i\ra \nu_i$ for some $\nu_i$ whenever $c_i$ is not a terminal; otherwise, let $\nu_i \define c_i$, and (2) $\nu = \sum_{i\in I} p_i\cdot \nu_i$. It is easy to check that such a $\nu$ is a valid distribution over configurations. 
	Let $\ra^k$ be the $k$-fold composition of $\ra$, and ${\ra^*} \define \bigcup_{k\geq 0} \ra^k$ the reflexive and transitive closure of $\ra$.

	Let $S\in \prog$, and $\<\cstate, \rho\>\in \qstatesh{V}$ with $V\supseteq \qv(S)$ and $\tr(\rho)=1$.
	A \emph{computation} of $S$ starting in $\<\cstate, \rho\>$ is an infinite sequence $\pi \define \{\mu_i: i\geq 0\}$ of distributions over configurations where $\mu_0 = \<S, \cstate, \rho\>$ and for each $i\geq 0$, $\mu_i \ra \mu_{i+1}$.

	\begin{lemma}\label{lem:infcom}
		Let $\pi \define \{\mu_i: i\geq 0\}$ be a computation starting in $\<\cstate, \rho\>$. Then $\qstate_{\mu_0}\le \qstate_{\mu_1} \le \ldots$.
	\end{lemma}

	With Lemma~\ref{lem:infcom}, we can define for any computation  $\pi \define \{\mu_i: i\geq 0\}$ the cq-state computed by $\pi$ as $\qstate_\pi\define\bigvee_{i\geq 0} \qstate_{\mu_i}$, the least upper bound of $\qstate_{\mu_i}$ according to $\le$.
	
	\begin{example}[Operational Semantics of Quantum Teleportation]\label{exa:ostelep}
		Let $\cstate$ be a classical state and $|\psi\>$ a pure state in $\h_2$. Then one of the computations, denoted $\pi$, of $\mathit{Teleport}$ starting in $\left\<\cstate, |\psi\>_q\<\psi|\otimes |\beta\>_{q_1, q_2}\<\beta|\right\>$ is shown as follows:
		\begin{eqnarray*}
			& &\conf{\mathit{Teleport}}{\cstate}{ [|\psi, \beta\>]}\\
			&\ra^5 & \sum_{i,j=0,1} \frac 14 \cdot  \conf{\mathbf{do}_a\| \mathbf{do}_b}{\cstate[i/x_A, j/z_A, 0/stage_A, 0/stage_B}{[|j,i, X^iZ^j\psi\>]}\\
			&\ra & \sum_{i,j=0,1} \frac 14 \cdot  \left\<stage_A := 1;\mathbf{do}_a\| stage_B := 1;  \iif\ {x_B=1} \ra {q_2\apply X}\ \square\ \neg (x_B =1)\ra  \right. \\
			& & \qquad\qquad\left. {\sskip};\ \mathbf{do}_b, \cstate[i/x_A, j/z_A, 0/stage_A, 0/stage_B, i/x_B], [|j,i, X^iZ^j\psi\>]\right\>\\
			&\ra^4 & \sum_{i,j=0,1} \frac 14 \cdot  \left\<\mathbf{do}_a\|  \mathbf{do}_b, \cstate[i/x_A, j/z_A, 1/stage_A, 1/stage_B, i/x_B], [|j,i, Z^j\psi\>]\right\>\\	
			&\ra^5 & \sum_{i,j=0,1} \frac 14 \cdot  \left\<\mathbf{do}_a\|  \mathbf{do}_b, \cstate[i/x_A, j/z_A, 2/stage_A, 2/stage_B, i/x_B, j/z_B], [|j,i, \psi\>]\right\>\\	
			&\ra^2 & \mu \define \sum_{i,j=0,1} \frac 14 \cdot  \left\<E, \cstate[i/x_A, j/z_A, 2/stage_A, 2/stage_B, i/x_B, j/z_B], [|j,i, \psi\>]\right\>\\
			&\ra & \mu  \ra \cdots
		\end{eqnarray*}
		where $\mathbf{do}_a$ and $\mathbf{do}_b$ are the $\mathbf{do}$-loops of $\textit{Alice}$ and $\textit{Bob}$, respectively. For pure state $|\phi\>$, we denote by $[|\phi\>]$ its corresponding density operator $|\phi\>\<\phi|$. Thus $$\qstate_\pi = \sum_{i,j=0,1}  \cqs{\cstate[i/x_A, j/z_A, 2/stage_A, 2/stage_B, i/x_B, j/z_B]}{\frac 14[|j,i, \psi\>]}.$$
	\end{example}
	
	Note that although each component process of a distributed program is deterministic, the whole program can still exhibit nondeterminism. This is due to the
	interleaving nature of local actions of individual processes and communication between disjoint pairs of processes; see Rules (Paral) and (Comm) in Table~\ref{tbl:opsemantics}. However, the following theorem shows that these different computations actually compute the same cq-state.

	\begin{theorem}[Determinism]\label{thm:determ}
		Let $S\in \prog$ be a distributed quantum program, and $\<\cstate, \rho\>\in \qstatesh{V}$ with $V\supseteq \qv(S)$ and $\tr(\rho)=1$. Then the set $$\{\qstate_\pi : \pi \mbox{ is a computation of $S$ starting in $\<\cstate, \rho\>$}\}$$ has exactly one element.
	\end{theorem}

	\subsection{Denotational Semantics}

	With Theorem~\ref{thm:determ}, the denotational semantics of distributed quantum programs can be defined using the operational one. Let $\qstatesh{\supseteq \qv(S)} \define\bigcup_{V\supseteq qv(S)} \qstatesh{V}$.
	
	\begin{definition}\label{def:denotational}
		Let $S\in \prog$. The \emph{denotational semantics} of $S$ is a mapping 
		$\sem{S} : \qstatesh{\supseteq \qv(S)}\ra  \qstatesh{\supseteq \qv(S)}$
		such that 
		\begin{enumerate}
			\item for any $\<\cstate, \rho\>\in  \qstatesh{V}$ with $V\supseteq qv(S)$ and $\tr(\rho) =1$,
			$$\sem{S}(\cstate, \rho) \define \mbox{the unique element in } \{\qstate_\pi : \pi \mbox{ is a computation of $S$ starting in $\<\cstate, \rho\>$}\};$$
			\item for any  $\qstate = \bigoplus_{i\in I}\<\cstate_i, \rho_i\>$ (thus $\tr(\rho_i) >0$ for any $i\in I$), 
			$$\sem{S}(\qstate) \define \sum_{i\in I}\tr(\rho_i) \cdot \sem{S}\left(\cstate_i, \frac{\rho_i}{\tr(\rho_i)}\right).$$
		\end{enumerate}
	\end{definition}
	
	To simplify notation, we always write $(\cstate, \rho)$ for  $(\<\cstate, \rho\>)$ when $\<\cstate, \rho\>$ appears as a parameter of some function. 
	The next lemma guarantees the well-definedness of Definition~\ref{def:denotational}.
	
	\begin{lemma}\label{lem:welldef}
		Let $S\in \prog$ and $\qstate\in \qstatesh{V}$ with $V\supseteq \qv(S)$. Then 
		\begin{enumerate}
			\item  $\sem{S}(\qstate)$ has countable support, and $\tr(\sem{S}(\qstate)) \leq \tr(\qstate)$. Hence $\sem{S}(\qstate) \in \qstates$ as well;
			\item for any $\lambda_i\in \R$, $\sem{S}(\qstate) = \sum_i \lambda_i \cdot \sem{S}(\qstate_i)$ whenever $\qstate = \sum_i \lambda_i\cdot\qstate_i$.
		\end{enumerate}
	\end{lemma}

	%	
	%	\begin{example}[Denotational Semantics  of Quantum Teleportation]
	%		The denotational semantics of $Teleport$ with the initial cq-state $\<\cstate, [|\psi, \beta\>]\>$ is
	%		$\sem{Teleport}(\cstate, [|\psi, \beta\>]) =  \qstate_\pi$ where $\pi$ is defined in Example~\ref{exa:ostelep}.
	%	\end{example}
	
	\section{Transformation to sequential quantum programs}\label{sec:sequentialisation}
	
	Throughout this section, we consider a distributed quantum program
	$S \define S_1\|\cdots\|S_n$ where for each $i$,
	\[
	S_i \define S_{i,0}; \mathbf{do}\ \square_{j=1}^{m_i} B_{i,j}; \alpha_{i,j}\rightarrow S_{i,j}\ \mathbf{od}.
	\]
	The transformation of $S$ into a sequential one follows the standard approach for classical (non-probabilistic) programs~\cite{apt2010verification}. 
	
	Let
	$
	\Gamma \define \{(i,j,k,\ell) : \alpha_{i,j} \mbox{ and } \alpha_{k,\ell} \mbox{ match, and } i< k\} 
	$. That is, $\Gamma$ collects all the pairs of generalised guards in the component processes which are able to communicate. The \emph{sequentialisation} of $S$ is defined as
	\begin{align*}
		T(S) \define\ & S_{1,0}; \ldots;  S_{n,0};\\
		&\mathbf{do}\ \square_{(i,j,k,\ell)\in \Gamma}\ B_{i,j}\wedge B_{k,\ell} \wedge B_i\rightarrow \mathit{Effect}( \alpha_{i,j}, \alpha_{k,\ell}); S_{i,j}; S_{k,\ell}\\
		&\mathbf{od}
	\end{align*}
	where $B_i \define \bigwedge_{(t,j,k,\ell)\in \Gamma, t<i} \neg (B_{t,j}\wedge B_{k,\ell})$. When $\Gamma$ is empty, we simply drop the $\mathbf{do}$ loop in the definition.

	%	Let
	%	$
	%	\Gamma \define \{(i,j,k,\ell) : \alpha_{i,j} \mbox{ and } \alpha_{k,\ell} \mbox{ match, and } i< k\}, 
	%	$
	%	and
	%	\begin{align*}
	%		T(S) \define\ & S_{1,0}; \ldots;  S_{n,0};\\
	%		&\mathbf{do}\ \square_{(i,j,k,\ell)\in \Gamma}\ B_{i,j}\wedge B_{k,\ell} \wedge B_i\rightarrow \mathit{Effect}( \alpha_{i,j}, \alpha_{k,\ell}); S_{i,j}; S_{k,\ell}\\
	%		&\mathbf{od}
	%	\end{align*}
	%	where $B_i \define \bigwedge_{(t,j,k,\ell)\in \Gamma, t<i} \neg (B_{t,j}\wedge B_{k,\ell})$,
	%	be the \emph{sequentialisation} of $S$ where $\Gamma$ collects all the pairs of generalised guards in the component processes which are able to communicate. When $\Gamma$ is empty, we simply drop the $\mathbf{do}$ loop in the definition. 
	
	Note that we introduce an additional condition $B_i$ here to guarantee that the resultant quantum program is deterministic (so that it can be described in the language presented in Sec.~\ref{sec:seq}). This is unnecessary for classical programs in~\cite{apt2010verification}, since verification of nondeterministic classical programs has been well investigated. 
	However, from Theorem~\ref{thm:determ} the nondeterministic choices in $S$ do not really matter in computing the final cq-state. Therefore, introducing the additional condition $B_i$ does not put any restriction on the expressiveness of the sequentialised program $T(S)$; this will be more rigorously shown with Theorem~\ref{thm:sequentialisation} below.

	It is obvious that  $S$ and $T(S)$ are not semantically equivalent: at least they have different conditions for termination. To see this, let
	\[
	\mathit{TERM} \define \bigwedge_{i=1}^n \bigwedge_{j=1}^{m_i} \neg B_{i,j},\quad 	\mathit{BLOCK} \define \bigwedge_{(i,j,k,\ell)\in \Gamma} \neg (B_{i,j}\wedge B_{k,\ell}) = \bigwedge_{(i,j,k,\ell)\in \Gamma} \neg (B_{i,j}\wedge B_{k,\ell}\wedge B_i).
	\]
	Then $S$ terminates iff $\mathit{TERM}$ holds while $T(S)$ terminates iff $\mathit{BLOCK}$ holds. Note that $\mathit{TERM} \ra \mathit{BLOCK}$ but generally the reverse direction is not true.
	
	The following theorem shows that $S$ and $T(S)$ are indeed equivalent conditioning on $\mathit{TERM}$.
	
	\begin{theorem}\label{thm:sequentialisation} For any cq-state  $\qstate\in \qstatesh{V}$ with $V\supseteq \qv(S)$,
		$
		\sem{S}(\qstate) = \sem{T(S)}(\qstate)|_{\mathit{TERM}},
		$
		the restriction of $\sem{T(S)}(\qstate)$ on the set of classical states $\cstate$ with $\cstate\models \mathit{TERM}$.
	\end{theorem}

	\begin{example}[Sequentialisation of Teleportation]\label{ex:seqtel}
		The sequentialisation of $Teleport$, denoted $T(\mathit{Teleport}$), is as follows: 
		\begin{align*}
			& q,q_1 \apply \textit{CNOT};\ q\apply H; \ z_A := \mymeas\ q;\  x_A := \mymeas\ q_1; \\
			& stage_A := 0;\ stage_B := 0;\\
			& \mathbf{do}\ stage_A = 0 \wedge stage_B = 0 \ra x_B := x_A;\\
			& \qquad \qquad stage_A := 1;\ stage_B := 1; \ \measstm{x_B=1}{q_2\apply X}{\sskip} \\
			&\ \  \square\  \ stage_A = 1 \wedge stage_B = 1 \ra z_B := z_A;\\
			& \qquad \qquad stage_A := 2;\  stage_B := 2; \ \measstm{z_B=1}{q_2\apply Z}{\sskip} \\
			& \mathbf{od}
		\end{align*}
		It is easy to see that 
		\[
		\sem{Teleport}(\qstate) = \sem{T(Teleport)}(\qstate) |_{stage_A \not\in \{0,1\} \wedge stage_B \not\in \{0,1\}}.
		\]
	\end{example}
	
	\section{Verification of distributed quantum programs}
	\label{sec:verification}
	
	The basic notion for verification of distributed quantum programs is classical-quantum assertion  from~\cite{feng2020quantum}.
	
	\subsection{Classical-quantum assertions} 
	
	Recall that  assertions for classical program states are usually represented as first order logic formulas over $\CVar$.
	For any classical assertion $\cassert$, denote by $\sem{\cassert} \define \{\cstate\in \Sigma : \cstate\models \cassert\}$ the set of classical states that satisfy $\cassert$. Two assertions $\cassert$ and $\cassert'$ are equivalent, written $\cassert \equiv \cassert'$, iff $\sem{\cassert} = \sem{\cassert'}$. Let $\phv$ be the set of Hermitian operators on $\h$ whose eigenvalues lie between 0 and 1.
	
	\begin{definition}
		Given $V\subseteq \QVar$, a \emph{classical-quantum assertion} (cq-assertion for short) $\qassert$ over $V$ is a function in $\cstates\rightarrow\phv$ such that  
		\begin{enumerate}
			\item the image set $\qassert(\Sigma)$ of $\qassert$ is countable;
			\item for each $M\in \qassert(\Sigma)$, the preimage $\qassert^{-1}(M)$ is definable by a classical assertion $\cassert$ in the sense that $\sem{\cassert} =  \qassert^{-1}(M)$.
		\end{enumerate}
	\end{definition}
	
	Denote by $\qVar(\qassert)$ the set $V$ of quantum variables in $\qassert$. We write $\bigoplus_{i\in I}\<\cassert_i, M_i\>$ instead of $\bigoplus_{i\in I}\<\sem{\cassert_i}, M_i\>$ for a cq-assertion  $\qassert$ whenever
	$\qassert(\cstates) = \{M_i : i\in I\}$ 
	and $\qassert^{-1}(M_i)=\sem{\cassert_i}$ for each $i\in I$. Note that this representation is not unique: the representative assertion $\cassert_i$ can be replaced by $\cassert'_i$ whenever $\cassert_i\equiv \cassert'_i$.
	Furthermore, the summand with zero operator $\z_{\h_V}$ is always omitted.
	In particular, when $\qassert(\cstates) = \{\z_\h, M\}$ or $\{M\}$ for some $M\neq \z_{\h_V}$, we simply denote $\qassert$ by $\<\cassert, M\>$ for some $\cassert$ with $\qassert^{-1}(M)=\sem{\cassert}$.

	Let $\qassertsh{V}$ be the set of all cq-assertions over $V$, and $\a$ the set of all cq-assertions.
	Again, we extend the L\"{o}wner order $\le_V$ for $\l(\h_V)$ pointwisely to $\a$ by letting $\qassert\le \qassert'$ iff $qv(\qassert) = qv(\qassert')$ and for all $\cstate\in \Sigma$, $\qassert(\cstate) \le_{qv(\qassert)} \qassert'(\cstate)$. 
	It is easy to see that 
	$\qasserts$ is also a pointed $\omega$-CPO under $\le$, with the least element being $\emptydis_{V}$. Furthermore, it has the largest element $\top_{V} \define \<\true, I_{\h_V}\>$. 	
	When $\qassert \le \qassert'$,  we denote by $\qassert' - \qassert$ the unique $\qassert''\in \qassertsh{qv(\qassert)}$ such that $\qassert'' + \qassert = \qassert'$. With these notions, summation and linear-sum of cq-assertions can be defined similarly as for cq-states. Let $V_1, V_2$ be two subsets of $\QVar$, and $\qassert_i\in \a_{V_i}$, $i=1,2$. We say $\qassert_1\lesssim \qassert_2$ whenever $\qassert_1\otimes I_{\h_{V_2\backslash V_1}} \le I_{\h_{V_1\backslash V_2}} \otimes \qassert_2$. Obviously, when restricted on some given set of quantum variables, $\lesssim$ coincides with $\le$. 
	
	Given a classical assertion $\cassert$, we denote by $\cassert \bowtie \sum_{i} \<\cassert_i, M_i\>$ the cq-assertion $\sum_{i} \<\cassert \bowtie \cassert_i, M_i\>$  (if it is valid) where $\bowtie$ can be any logic connective such as $\wedge$, $\vee$, $\Rightarrow$, $\Leftrightarrow$, etc. Let $\f$ be a completely positive and sub-unital linear map from $\phv$ to $\mathcal{P}(\h_W)$. We extend it to $\qassertsh{V}$ in a pointwise way. In particular, when $qv(\qassert) \cap W = \emptyset$, $\qassert\otimes I_{\h_W}$ is a cq-assertion which maps any $\cstate\in \cstates$ to $\qassert(\cstate)\otimes I_{\h_W}$.

	\begin{definition}\label{def:satisfaction}
		Given a cq-state $\qstate$ and a cq-assertion $\qassert$ with $\qv(\qstate) \supseteq \qv(\qassert)$, the \emph{expectation} of $\qstate$ satisfying $\qassert$ is defined to be 
		\[
		\Exp(\qstate \models \qassert) \define 
		\sum_{\cstate\in \supp{\qstate}} \tr\left[\left(\qassert(\cstate)\otimes I_{\h_{V}} \right) \cdot \qstate(\cstate)\right]
		= 
		\sum_{\cstate\in \supp{\qstate}} \tr\left[\qassert(\cstate) \cdot \tr_{\h_{V}}(\qstate(\cstate))\right]
		\] 
		where $V = \qv(\qstate) \backslash \qv(\qassert)$ and the dot $\cdot$ denotes matrix multiplication.
	\end{definition}
	
	\subsection{Correctness formula}
	
	As usual, program correctness is expressed by \emph{correctness formulas} with the form
	$\ass{\qassert}{S}{\qassertp}$
	where $S$ is a distribute quantum program, and $\qassert$ and $\qassertp$ are both cq-assertions.
	We do not put any requirement on the quantum variables which $\qassert$ and $\qassertp$ are acting on. In fact, the sets $\qv(S)$, $\qv(\qassert)$, and $\qv(\qassertp)$ can be all different.

	\begin{definition}
		Let $S\in \prog$, and $\qassert$ and $\qassertp$ be cq-assertions.
		\begin{enumerate}
			\item We say the correctness formula $\ass{\qassert}{S}{\qassertp}$ is true in the sense of \emph{total correctness}, written $\models_{\tot} \ass{\qassert}{S}{\qassertp}$, if for any 
			$V\supseteq \qv(S, \qassert, \qassertp)$ and
			$\qstate\in \qstatesh{V}$,
			$$\Exp(\qstate\models \qassert) \leq \Exp(\sem{S}(\qstate) \models \qassertp).$$
			\item We say the correctness formula $\ass{\qassert}{S}{\qassertp}$ is true in the sense of \emph{partial correctness}, written $\models_{\pal} \ass{\qassert}{S}{\qassertp}$, if for any 
			$V\supseteq \qv(S, \qassert, \qassertp)$ and
			$\qstate\in \qstatesh{V}$, 
			$$\Exp(\qstate\models \qassert) \leq \Exp(\sem{S}(\qstate) \models \qassertp)+  \tr(\qstate) - \tr(\sem{S}(\qstate)).$$
		\end{enumerate}
	\end{definition}
	
	%	
	%	The next lemma shows that the validity of correctness formulas can be checked on simple cq-states.
	%	\begin{lemma}\label{lem:formulasimple}
	%		Let $S\in \prog$, $\qassert$ and $\qassertp$ be cq-assertions, and $V\define \qv(S, \qassert, \qassertp)$. Then 
	%		\begin{enumerate}
	%			\item  $\models_{\tot} \ass{\qassert}{S}{\qassertp}$ iff for any 
	%			$\<\cstate, \rho\> \in \qstatesh{V}$ with $\tr(\rho) = 1$,
	%			$$\Exp(\<\cstate, \rho\>\models \qassert) \leq \Exp(\sem{S}(\cstate, \rho)\models \qassertp).$$
	%			\item $\models_{\pal} \ass{\qassert}{S}{\qassertp}$ iff for any 
	%			$\<\cstate, \rho\> \in \qstatesh{V}$ with $\tr(\rho) =1$,
	%			$$\Exp(\<\cstate, \rho\>\models \qassert) \leq \Exp(\sem{S}(\cstate, \rho) \models \qassertp)+  \tr(\rho) - \tr(\sem{S}(\cstate, \rho)).$$
	%		\end{enumerate}
	%	\end{lemma}
	%	\begin{proof}
	%		Easy from linearity of $\sem{S}$ for any program $S$; see Lemma~\ref{lem:welldef}(2).	
	%	\end{proof}
	
	\begin{example}
		The correctness of quantum teleportation can be stated as  follows: for any $|\psi\>\in \h_2$,
		\[
		\models_{\tot}\ass{|\psi\>_q\otimes |\beta\>_{q_1,q_2}}{\mathit{Teleport}}{|\psi\>_{q_2}},
		\]
		which claims that the (arbitrary) quantum state of qubit $q$ is successfully
		transmitted to qubit $q_2$ by $\mathit{Teleport}$. 
		Note that the postcondition $|\psi\>_{q_2}$ does not refer to $q$ and $q_1$, meaning that the post-measurement state of these quantum systems is irrelevant. 
	\end{example}
	
	%	
	%	Finally, we show some basic facts about total and partial correctness as follows.
	%	
	%	\begin{lemma}\label{lem:corf}
	%		Let $S\in \prog$, $\qassert$ and $\qassertp$ be cq-assertions, and $V\subseteq \QVar$. 
	%		\begin{enumerate}
	%			\item If $\models_{\tot} \ass{\qassert}{S}{\qassertp}$ then $\models_{\pal} \ass{\qassert}{S}{\qassertp}$;
	%			\item $\models_{\tot} \ass{\emptydis_{V}}{S}{\qassertp}$ and $\models_{\pal} \ass{\qassert}{S}{\top_{V}}$;
	%			\item If $\models_{\tot} \ass{\qassert_i}{S}{\qassertp_i}$ and  $\lambda_i\geq 0$ for $i=1,2$, then 
	%			$
	%			\models_{\tot} \ass{ \lambda_1\qassert_1 +  \lambda_2\qassert_2}{S}{ \lambda_1 \qassertp_1 + \lambda_2 \qassertp_2}.
	%			$
	%			The result also holds for partial correctness if $\lambda_1 +  \lambda_2 =1$.
	%		\end{enumerate}
	%	\end{lemma}
	
	\subsection{Proof systems}

	{\renewcommand{\arraystretch}{2.5}
		\begin{table}[t]
			\begin{lrbox}{\tablebox}
				\centering
				\begin{tabular}{l}
					\begin{tabular}{lclc}
						(Skip)	& $\ass{\qassert}{\sskip}{\qassert}$ & 
						(Abort) 	& $\ass{\top_{V}}{\abort}{\bot_V}$\\
						(Assn)	&
						$\ass{\qassert[\subs{x}{e}]}{x:=e}{\qassert}$ &
						(Rassn) &
						$\displaystyle\left\{\sum_{d\in D_{\mathit{type}(x)}} g(d)\cdot \qassert[\subs{x}{d}]\right\}{x\rassign g}\{\qassert\}$ \\
						(Init)	& $\displaystyle\frac{q\in \qv(\qassert)}{\ass{\sum_{i=0}^{d_q-1} \qiz \qassert\qzi}{q:=0}{\qassert}}$ &
						(Unit)	&
						$\displaystyle\frac{\bar{q}\subseteq \qv(\qassert)}{\ass{U_{\bar{q}}^\dag \qassert U_{\bar{q}}}{\bar{q}\apply U}{\qassert}}$ \\
						(Meas)	&
						$\displaystyle\frac{\bar{q}\subseteq \qv(\qassert), \m = \{M_i : i\in I\}}{\ass{\sum_{i\in I}M_i^\dag\qassert[\subs{x}{i}]M_i}{x:=\mymeas\ \m[\bar{q}]}{\qassert}}$  &
						(Seq)	&
						$\displaystyle\frac{\ass{\qassert}{S_0}{\qassert'},\ \ass{\qassert'}{S_1}{\qassertp}}{\ass{\qassert}{S_0; S_1}{\qassertp}}$\\
						(Alt)	&
						$\displaystyle\frac{\ass{B_i\wedge \qassert}{S_i}{\qassertp},\ \forall i\in \{1,\ldots, n\}}{\ass{\qassert}{\altercom}{\qassertp}}$
						&
						(Rep)	& $\displaystyle\frac{\ass{B_i\wedge \qassert}{S_i}{\qassert},\ \forall i\in \{1,\ldots, n\}}{\ass{\qassert}{\repcom}{\qassert \wedge \bigwedge_{i=1}^n \neg B_i}}$ \\
						(Imp)	&
						$\displaystyle\frac{\qassert\lesssim \qassert',\ \ass{\qassert'}{S}{\qassertp'},\ \qassertp'\lesssim \qassertp}{\ass{\qassert}{S}{\qassertp}}$& &
					\end{tabular}\\
					\begin{tabular}{lc}
						(Dist) & $\displaystyle\frac{
							\ass{\qassert}{S_{1,0}; \ldots;  S_{n,0}}{\qassertp}, \ 				\ass{B_{i,j}\wedge B_{k,\ell}\wedge\qassertp}{\mathit{Effect}(\alpha_{i,j}, \alpha_{k,\ell}); S_{i,j}; S_{k,\ell}}{\qassertp}, \forall (i,j,k,\ell)\in \Gamma
						}
						{\ass{\qassert}{S_1\|\ldots\|S_n}{\qassertp\wedge \mathit{TERM}}}$\\
						& where $\Gamma$ and $\mathit{TERM}$ are defined as in Sec.~\ref{sec:sequentialisation}.
					\end{tabular}		
				\end{tabular}
			\end{lrbox}
			\resizebox{\textwidth}{!}{\usebox{\tablebox}}\\
			\vspace{4mm}
			\caption{Proof system for partial correctness. 
			}
			\label{tbl:psystem}
		\end{table}
	}

	The core of Hoare logic is a proof system consisting of axioms and proof rules which enable syntax-oriented and modular reasoning of program correctness. In this section, we propose a Hoare logic for distributed quantum programs.
	
	\textbf{Partial correctness}. We propose in Table~\ref{tbl:psystem} a proof system for partial correctness of distributed quantum programs, which is a natural extension of the quantum Hoare logic introduced in~\cite{feng2020quantum} for deterministic while programs. We write $\vdash_{\pal}\ass{\qassert}{S}{\qassertp}$ if the correctness formula $\ass{\qassert}{S}{\qassertp}$ can be derived from the system.

	\begin{theorem}\label{thm:psc}
		The proof system in Table~\ref{tbl:psystem} is both sound and (relatively)  complete with respect to the partial correctness of distributed quantum programs.
	\end{theorem}

	\textbf{Total correctness}. 
	Ranking functions play a central role in proving total correctness of while loop programs. Recall that in the classical case, a ranking function maps each reachable state in the loop body to an element of a well-founded ordered set (say, the set $\N$ of nonnegative integers), such that the value decreases strictly after each iteration of the loop. Our proof rules for total correctness of repetitive commands and distributed quantum programs also heavily relies on the notion of ranking assertions.

	\begin{definition}
		Let $\qassert\in \qasserts$. A decreasing sequence (w.r.t. $\le$) of cq-assertions $\{\qassert_k : k\geq 0\}$ in $\qasserts$ with  $\qassert \le \qassert_0$ and $\bigwedge_k \qassert_k = \emptydis_V$ are $\qassert$-\emph{ranking assertions} for $\repcom$ if for any $k\geq 0$, $1\leq i\leq n$, and $\qstate\in \qstatesh{W}$, $W\define \bigcup_{i=1}^n qv(S_i)\cup V$,
		\begin{equation}\label{eq:rankfun}
			\Exp(\sem{S_i}(\qstate|_{B_i}) \models \qassert_k ) \leq \Exp(\qstate \models \qassert_{k+1}).
		\end{equation}
		They are said to be $\qassert$-\emph{ranking assertions} for $S_1\|\ldots\|S_n$ if, for any $k\geq 0$, $(i,j,t,\ell)\in \Gamma$, and $\qstate\in \qstatesh{W}$, $W\define \bigcup_{i=1}^n qv(S_i)\cup V$, we have
		\[
		\Exp(\sem{S_{i,j}^{t,\ell}}(\qstate|_{B_{i,j}\wedge B_{t,\ell}}) \models \qassert_k ) \leq \Exp(\qstate \models \qassert_{k+1})
		\]
		where $S_{i,j}^{t,\ell} \define \mathit{Effect}(\alpha_{i,j}, \alpha_{t,\ell}); S_{i,j}; S_{t,\ell}$.
	\end{definition}
	%	Note that $\qassert$-ranking assertions can also be defined by employing the weakest precondition semantics instead of the denotational one. For example, for $\repcom$, the defining equation~\eqref{eq:rankfun} can be replaced by $\sum_{i=1}^{n}B_i\wedge wp.S_i.\qassert_k \le \qassert_{k+1},$ thanks to the fact that $B_i$'s are mutually exclusive. It is easy to show that these two definitions are equivalent. 

	{\renewcommand{\arraystretch}{1}
		\begin{table}[t]
			\begin{lrbox}{\tablebox}
				\centering
				\begin{tabular}{lc}
					(Abort-T) & $\ass{\bot_{V}}{\abort}{\bot_{V}}$\\
					\\
					(Alt-T)	&
					$\displaystyle\frac{\qassert \lesssim \bigvee_{i=1}^n B_i, \ \ass{B_i\wedge \qassert}{S_i}{\qassertp},\ \forall i\in \{1,\ldots, n\}}{\ass{\qassert}{\altercom}{\qassertp}}$\\
					\\
					(Rep-T)	& $\displaystyle\frac{
						\begin{tabular}{l}
							$\ass{B_i\wedge \qassert}{S_i}{\qassert},\ \forall i\in \{1,\ldots, n\}$\\
							$\qassert$-ranking assertions exist for $\repcom$	
						\end{tabular}	
					}{\ass{\qassert}{\repcom}{\qassert \wedge \bigwedge_{i=1}^n \neg B_i}}$ \\
					\\
					(Dist-T) & $\displaystyle\frac{
						\begin{tabular}{l}
							$\ass{\qassert}{S_{1,0}; \ldots;  S_{n,0}}{\qassertp}$, and $\qassertp$-ranking assertions exist for $S_1\|\ldots\|S_n$\\ $\ass{B_{i,j}\wedge B_{k,\ell}\wedge\qassertp}{\mathit{Effect}(\alpha_{i,j}, \alpha_{k,\ell}); S_{i,j}; S_{k,\ell}}{\qassertp}, \forall (i,j,k,\ell)\in \Gamma$\\
							$\qassertp \wedge \mathit{BLOCK} \lesssim \mathit{TERM}$	
						\end{tabular}
					}	
					{\ass{\qassert}{S_1\|\ldots\|S_n}{\qassertp\wedge \mathit{TERM}}}$\\ \\
					& where $\Gamma$ and $\mathit{TERM}$ are defined as in Sec.~\ref{sec:sequentialisation}.
				\end{tabular}		
			\end{lrbox}
			\resizebox{.85\textwidth}{!}{\usebox{\tablebox}}\\
			\vspace{4mm}
			\caption{Some proof rules for total correctness. 
			}
			\label{tbl:tsystem}
		\end{table}
	}

	The {proof system for total correctness} is then defined as for partial correctness, except that the rules (Abort), (Alt), (Rep), and (Dist) are replaced by their corresponding total correctness version shown in Table~\ref{tbl:tsystem}.
	We write $\vdash_{\tot}\ass{\qassert}{S}{\qassertp}$ if the correctness formula $\ass{\qassert}{S}{\qassertp}$ can be derived using this proof system. 
	
	\begin{theorem}\label{thm:total}
		The proof system for total correctness is both sound and (relatively) complete with respect to the total correctness of distributed quantum programs.
	\end{theorem}

	\section{Conclusion and future works}

	In this paper, we propose a distributed programming language for the purpose of formal description and verification of distributed quantum systems. A Hoare-style logic, which turns out to be sound and (relatively) complete for both partial and total correctness,  is introduced to help analysis of quantum programs written in this language. Effectiveness of the logic is demonstrated by its application in verification of quantum teleportation and local implementation of non-local CNOT gates, two important protocols widely used in distributed quantum systems.
	
	The distributed language investigated in this paper only allows local quantum operations and classical communication (LOCC). Although LOCC is a widely used quantum communication model, there are also important quantum communication protocols, such as Quantum Key Distribution~\cite{bennett1984quantum} and Quantum Leader Election~\cite{tani2005exact}, which do require transmission of quantum states. It is well known that this kind of quantum communication can be achieved by employing the teleportation protocol (provided that enough entanglement is pre-shared between relevant parties), and thus in principle these  protocols can be verified using the logic presented in this paper, but their verification in this way will be clumsy and inconvenient. Therefore, it is desirable  to extend our language to include quantum communication in future works. 
	To this end, we have to trace the ownership of each quantum system so that the no-cloning property~\cite{wootters1982single} of quantum information is not violated. We expect that the verification of such distributed quantum programs will be much more challenging.
	
	%	\subsection{Quantum Teleportation}

	%	
	%	\subsection{Teleporting a Qudit}
	%	
	%	The above protocol can be straightforwardly generalised for teleporting a $d$-level quantum state, called a qudit \cite{Bennett93, Werner01}. We can use $|0\rangle,|1\rangle,\ldots,|d-1\rangle$ to denote the basis states of a $d$-level quantum systems. Then a general state of this system can be written as a superposition of its basis states:
	%	$|\psi\rangle=\sum_{i=0}^{d-1}\alpha_i|i\rangle,$ where complex numbers $\alpha_i$ satisfy the normalisation condition: $\sum_{i=0}^{d-1}|\alpha_i|^2=1.$
	%	
	%	\subsection{Multipartite Teleportation}
	%	
	%	\subsection{Quantum Repeaters}
	
%	\bibliographystyle{abbrv}
	\bibliography{ref}
	
	\newpage
	\appendix
	
	\section{Preliminaries}
	
	This section is devoted to fixing some notations from linear algebra and quantum mechanics that will be used in this paper. For a thorough introduction of relevant backgrounds, we refer to~\cite[Chapter 2]{nielsen2002quantum}.
	
	\subsection{Basic linear algebra}
	Let $\h$ be a Hilbert space. In the finite-dimensional case which we are concerned with here, it is merely a complex linear space equipped with an inner product. Consequently, it is isomorphic to $\C^d$ where $d=\dim(\h)$, the dimension of $\h$.
	Following the tradition in quantum computing, vectors in $\h$ are denoted in the Dirac form $|\psi\>$. The inner product of $|\psi\>$ and $|\phi\>$ is written $\<\psi|\phi\>$, and they are \emph{orthogonal} if $\<\psi|\phi\> = 0$. The \emph{outer product} of them, denoted $|\psi\>\<\phi|$, is a rank-one linear operator which maps any $|\psi'\>$ in $\h$ to $\<\phi|\psi'\> |\psi\>$.
	The \emph{length} of $|\psi\>$ is defined to be $\||\psi\>\| \define\sqrt{\<\psi|\psi\>}$ and it is called \emph{normalised} if $\||\psi\>\|=1$. A set of vectors $B\define\{|i\> : i\in I\}$ in $\h$ is \emph{orthonormal} if each $|i\>$ is normalised and every two of them are orthogonal. Furthermore, if they span the whole space $\h$; that is, any vector in $\h$ can be written as a linear combination of vectors in $B$, then $B$ is called an \emph{orthonormal basis} of $\h$. 
	
	Let $\lh$ be the set of linear operators on $\h$, and $\z_\h$ and $I_\h$ the zero and identity operators respectively. Let $A\in \lh$. The \emph{trace} of $A$ is defined to be $\tr(A) \define \sum_{i\in I} \<i|A|i\>$ for some (or, equivalently, any) orthonormal basis $\{|i\> : i\in I\}$ of $\h$. The \emph{adjoint} of $A$, denoted $A^\dag$, is the unique linear operator in $\lh$ such that $\<\psi|A|\phi\> = \<\phi|A^\dag |\psi\>^*$ for all $|\psi\>, |\phi\>\in \h$. Here for a complex number $z$, $z^*$ denotes its conjugate. Operator $A$ is said to be \emph{normal} if $A^\dag  A = A A^\dag$, \emph{hermitian} if $A^\dag = A$, \emph{unitary} if $A^\dag A = I_\h$, and \emph{positive} if for all $|\psi\>\in \h$, $\<\psi|A|\psi\>\geq 0$. Obviously,  hermitian operators are normal, and both unitary operators and positive ones are hermitian. Any normal operator $A$ can be written into a \emph{spectral decomposition} form $A  = \sum_{i\in I} \lambda_i |i\>\<i|$ where $\{|i\> : i\in I\}$ constitute some orthonormal basis of $\h$. Furthermore, if $A$ is hermitian, then all $\lambda_i$'s are real; if $A$ is unitary, then all $\lambda_i$'s have unit length; if $A$ is positive, then all $\lambda_i$'s are non-negative.  The L\"owner (partial) order $\le_\h$ on the set of hermitian operators on $\h$ is defined by letting $A\le_\h B$ iff $B-A$ is positive. 
	
	Let $\h_1$ and $\h_2$ be two finite dimensional Hilbert spaces, and $\h_1\otimes \h_2$ their tensor product. 
	Let $A_i\in \l(\h_i)$. The tensor product of $A_1$ and $A_2$, denoted $A_1\otimes A_2$ is a linear operator in $\l(\h_1\otimes \h_2)$ such that
	$(A_1\otimes A_2)|(\psi_1\>\otimes |\psi_2)\> = (A_1|\psi_1\>)\otimes (A_2|\psi_2\>)$ for all $|\psi_i\> \in \h_i$. To simplify notations, we often write $|\psi_1\> |\psi_2\>$ for $|\psi_1\>\otimes |\psi_2\>$.
	Given $\h_1$ and $\h_2$, the \emph{partial trace} with respect to $\h_2$, denoted $\tr_{\h_2}$, is a linear mapping from
	$\l(\h_1\otimes \h_2)$ to $\l(\h_1)$ such that for any $|\psi_i\>, |\phi_i\> \in \h_i$, $i=1,2$,
	$$\tr_{\h_2}(|\psi_1\>\<\phi_1|\otimes |\phi_1\>\<\phi_2|) = 
	\<\phi_2|\phi_1\> |\psi_1\>\<\phi_1|.$$
	The definition is extended to $\l(\h_1\otimes \h_2)$ by linearity.
	
	A linear operator $\e$ from $\l(\h_1)$ to $\l(\h_2)$ is called a \emph{super-operator}.  It is said to be (1) \emph{positive} if it maps positive operators to positive operators; (2) \emph{completely positive} if all the cylinder extension $\mathcal{I}_\h\otimes \e$ is positive for all finite dimensional Hilbert space $\h$, where $\mathcal{I}_\h$ is the identity super-operator on $\lh$; (3) \emph{trace-preserving} (resp. \emph{trace-nonincreasing}) if 
	$\tr(\e(A)) = \tr(A)$ (resp. $\tr(\e(A)) \leq \tr(A)$ for any positive operator $A\in \l(\h_1)$; (4) \emph{unital} (resp. \emph{sub-unital}) if 
	$\e(I_{\h_1})= I_{\h_2}$ (resp. $\e(I_{\h_1}) \le_{\h_2} I_{\h_2}$).
	From \emph{Kraus representation theorem}~\cite{kraus1983states}, a super-operator $\e$  from $\l(\h_1)$ to $\l(\h_2)$ is completely positive iff there is some set of linear operators, called \emph{Kraus operators}, $\{E_i : i\in I\}$ from $\h_1$ to $\h_2$ such that $\e(A) = \sum_{i\in I} E_i A E_i^\dag$ for all $A\in \l(\h_1)$. 
	%Then $\e$ is further trace-preserving (resp. trace-nonincreasing) iff 
	%$\sum_i E_i^\dag E_i = I_{\h_1}$ (resp. $\le I_{\h_1}$), and it is unital (resp. sub-unital) iff $\sum_i E_i E_i^\dag = I_{\h_2}$ (resp. $\le I_{\h_2}$).
	It is easy to check that the trace and partial trace operations defined above are both completely positive and trace-preserving super-operators. 
	Given a completely positive super-operator $\e$ from $\l(\h_1)$ to $\l(\h_2)$ with Kraus operators $\{E_i : i\in I\}$, the adjoint of $\e$, denoted $\e^\dag$, is a completely positive super-operator from $\l(\h_2)$ back to $\l(\h_1)$ with Kraus operators $\{E_i^\dag : i\in I\}$. Then we have $(\e^\dagger)^\dag = \e$, and $\e$ is trace-preserving (resp. trace-nonincreasing) iff $\e^\dag$ is unital (resp. sub-unital). Furthermore, for any $A\in \l(\h_1)$ and $B\in \l(\h_2)$, $\tr(\e(A)\cdot B) = \tr(A\cdot \e^\dag(B))$.

	\subsection{Basic quantum mechanics}
	
	According to von Neumann's formalism of quantum mechanics
	\cite{vN55}, any quantum system with finite degrees of freedom is associated with a finite-dimensional Hilbert space $\h$ called its \emph{state space}. When $\dim(\h) = 2$, we call such a system a \emph{qubit}, the analogy of bit in classical computing. A {\it pure state} of the system is described by a normalised vector in $\h$. When the system is in one of an {ensemble} of states $\{|\psi_i\>: i\in I\}$ with respective probabilities $p_i$, we say it is in a \emph{mixed} state, represented by the \emph{density operator} $\sum_{i\in I} p_i|\psi_i\>\<\psi_i|$ on $\h$. Obviously, a density operator is positive and has trace 1. Conversely, by spectral decomposition, any positive operator with unit trace corresponds to some (not necessarily unique) mixed state.
	
	The state space of a composite system (for example, a quantum system
	consisting of multiple qubits) is the tensor product of the state spaces
	of its components. For a mixed state $\rho$ in $\h_1 \otimes \h_2$,
	partial traces of $\rho$ have explicit physical meanings: the
	density operators $\tr_{\h_1}(\rho)$ and $\tr_{\h_2}(\rho)$ are exactly
	the reduced quantum states of $\rho$ on the second and the first
	component systems, respectively. Note that in general, the state of a
	composite system cannot be decomposed into tensor product of the
	reduced states on its component systems. A well-known example is the
	2-qubit state
	$|\Psi\>=\frac{1}{\sqrt{2}}(|00\>+|11\>).
	$
	This kind of state is called {\it entangled state}, and usually is the key to many quantum information processing tasks  such as teleportation
	\cite{bennett1993teleporting} and superdense coding \cite{bennett1992communication}.
	
	The \emph{evolution} of a closed quantum system is described by a unitary
	operator on its state space: if the states of the system at times
	$t_1$ and $t_2$ are $\rho_1$ and $\rho_2$, respectively, then
	$\rho_2=U\rho_1U^{\dag}$ for some unitary operator $U$ which
	depends only on $t_1$ and $t_2$. In contrast, the general dynamics which can occur in a physical system is
	described by a completely positive and trace-preserving super-operator on its state space. 
	Note that the unitary transformation $\e_U(\rho)\define U\rho U^\dag$ is
	such a super-operator. 
	
	A quantum {\it measurement} $\m$ is described by a
	collection $\{M_i : i\in I\}$ of linear operators on $\h$, where $I$ is the set of measurement outcomes. It is required that the
	measurement operators satisfy the completeness equation
	$\sum_{i\in I}M_i^{\dag}M_i = I_\h$. If the system is in state $\rho$, then the probability
	that measurement result $i$ occurs is given by
	$p_i=\tr(M_i^{\dag}M_i\rho),$ and the state of the post-measurement system
	is $\rho_i = M_i\rho M_i^{\dag}/p_i$ whenever $p_i>0$. 
	Note that the super-operator $$\e_\m: \rho\mapsto
	\sum_{i\in I} p_i \rho_i = \sum_{i\in I} M_i\rho M_i^\dag$$
	which maps the initial state to the final (mixed) one when the measurement outcome is ignored is completely positive and trace-preserving.
	A particular case of measurement is {\it projective measurement} which is usually represented by a hermitian operator $M$ in $\lh$ called \emph{observable}.  Let 
	\[
	M=\sum_{m\in \mathit{spec}(M)}mP_m
	\] 
	where $\mathit{spec}(M)$ is the set of eigenvalues of $M$, and $P_m$ the projection onto the eigenspace associated with $m$. 
	Obviously, the projectors  $\{P_m:m\in
	spec(M)\}$ form a quantum measurement. 
	
	In this paper, we are especially concerned with the set 
	\[
	\ph \define \{M\in \lh : \z_\h\le M\le I_\h\}
	\]
	of observables whose eigenvalues lie between 0 and 1, where $\le$ is the L\"owner  order on $\lh$. Furthermore, following Selinger's convention~\cite{selinger2004towards}, we regard the set of \emph{partial density operators} 
	\[
	\dh \define \{\rho\in \lh : \z_\h\le \rho, \tr(\rho)\leq 1\}
	\]
	as (unnormalised) quantum states. Intuitively, the partial density operator $\rho$ means that the legitimate quantum state $\rho/\tr(\rho)$ is reached with probability $\tr(\rho)$.
	As a matter of fact, we note that $\dh\subseteq\ph$.
	
	\section{Some useful lemmas}
	
	We first recall some basic properties of cq-states and cq-assertions from~\cite{feng2020quantum}.
	
	\begin{lemma}[Lemma 3.9, \cite{feng2020quantum}]\label{lem:bpdeg}
		For any cq-state $\qstate\in \qstatesh{V}$, cq-assertion $\qassert\in \qassertsh{W}$ with $W\subseteq V$, and classical assertion $\cassert$,
		\begin{enumerate}
			\item $\Exp(\qstate \models \qassert)\in [0,1]$;
			\item $\Exp(\emptydis_V \models \qassert) = \Exp(\qstate \models \emptydis_{W})=0$, $\Exp(\qstate \models \top_{W}) = \tr(\qstate)$;
			\item $\Exp(\qstate \models \qassert) = \sum_i \lambda_i \Exp(\qstate \models \qassert_i)$ if $\qassert = \sum_i \lambda_i \qassert_i$;
			\item $\Exp(\qstate \models \qassert) = \sum_i \lambda_i \Exp(\qstate_i \models \qassert)$ if $\qstate = \sum_i \lambda_i \qstate_i$;
			\item\label{cl:lem3.9} $\Exp(\qstate |_\cassert \models \qassert) = \Exp(\qstate \models \cassert\wedge\qassert)$;
			{\item\label{cl:super} $\Exp(\qstate \models \f(\qassertp)) = \Exp(\f^\dag(\qstate) \models \qassertp)$ for any $\qassertp\in \qassertsh{W'}$ and any completely positive and sub-unital super-operator $\f$ from $\h_{W'}$ to $\h_W$.}
		\end{enumerate} 
	\end{lemma}
	
	\begin{lemma}[Lemma 3.10, \cite{feng2020quantum}]\label{lem:qassetorder}\quad
		\begin{enumerate}
			\item For any cq-states $\qstate$ and $\qstate'$ in $\qstatesh{V}$,
			\begin{itemize}
				\item if $\qstate \le \qstate'$, then $\Exp(\qstate \models \qassert)\leq \Exp(\qstate' \models \qassert)$ for all $\qassert\in \qassertsh{W}$ with $W\subseteq V$;
				\item conversely, if $\Exp(\qstate \models \qassert)\leq \Exp(\qstate' \models \qassert)$ for all $\qassert\in \qassertsh{V}$, then $\qstate \le \qstate'$.
			\end{itemize} 
			\item For any cq-assertions $\qassert$ and $\qassert'$ with $W=qv(\qassert)\cup qv(\qassert')$,
			\begin{itemize}
				\item if $\qassert \lesssim \qassert'$, then $\Exp(\qstate \models \qassert)\leq \Exp(\qstate \models \qassert')$ for all $\qstate\in \qstatesh{V}$ with $W\subseteq V$;
				\item conversely, if $\Exp(\qstate \models \qassert)\leq \Exp(\qstate \models \qassert')$ for all $\qstate\in \qstatesh{W}$, then $\qassert \lesssim \qassert'$.
			\end{itemize} 
		\end{enumerate} 
	\end{lemma}

	\begin{lemma}[Lemma 3.11, \cite{feng2020quantum}]\label{lem:qasset}
		For any cq-states $\qstate, \qstate_n\in \qstatesh{V}$ and cq-assertions $\qassert, \qassert_{n} \in \qassertsh{W}$ with $W\subseteq V$, $n=1, 2, \cdots$,
		\begin{enumerate}
			\item $\Exp(\bigvee_{n\geq 0}\qstate_n \models \qassert) = \sup_{n\geq 0}\Exp(\qstate_n\models \qassert)$ for increasing sequence $\{\qstate_n\}_n$;		
			\item $\Exp(\bigwedge_{n\geq 0}\qstate_n \models \qassert) = \inf_{n\geq 0}\Exp(\qstate_n\models \qassert)$ for decreasing sequence $\{\qstate_n\}_n$;
			\item $\Exp(\qstate \models \bigvee_{n\geq 0}\qassert_n) = \sup_{n\geq 0}\Exp(\qstate\models \qassert_n)$ for increasing sequence $\{\qassert_n\}_n$;
			\item $\Exp(\qstate \models \bigwedge_{n\geq 0}\qassert_n) = \inf_{n\geq 0}\Exp(\qstate\models \qassert_n)$ for decreasing sequence $\{\qassert_n\}_n$.	
		\end{enumerate} 
	\end{lemma}

	The following lemma presents the explicit form for denotational semantics of various constructs for sequential programs, which extends 
	\cite[Lemma 4.6]{feng2020quantum}.
	
	\begin{lemma}\label{lem:iout} For any cq-states $\<\cstate, \rho\>$ and $\qstate$ in $\s_V$ where $V$ contains all quantum variables of the corresponding program,
		\begin{enumerate}
			%\item $\sem{S}$ is a trace-nonincreasing CP map.
			\item  $\sem{\sskip}(\qstate) = \qstate$, $\sem{\abort}(\qstate) = \bot_{V}$;
			\item  $\sem{x:=e}(\cstate, \rho) = \<\cstate[\subs{x}{\cstate(e)}], \rho\>$;
			\item  $\sem{x\rassign g}(\cstate, \rho) = \sum_{d\in D_{\mathit{type}(x)}} \<\cstate[\subs{x}{d}], g(d) \cdot \rho\>$;
			\item $\sem{x:= \mymeas\ \m[\bar{q}]}(\cstate, \rho) = \sum_{i\in I}  \<\cstate[\subs{x}{i}], M_i \rho M_i^\dag\>$ where $M_i$'s are applied on $\bar{q}$, and $\m=\{M_i : i\in I\}$;
			\item  $\sem{q:=0}(\cstate, \rho) = \<\cstate, \sum_{i=0}^{d_q-1}\qzi \rho\qiz\>\>$;
			\item $\sem{\bar{q}\apply U}(\cstate, \rho) = \<\cstate, U_{\bar{q}} \rho U_{\bar{q}}^\dag\>$.
			\item $\sem{S_0; S_1}(\qstate) = \sem{S_1}(\sem{S_0}(\qstate))$;
			\item $\sem{\altercom}(\qstate) = \sum_{i=1}^{n}\sem{S_i}(\qstate|_{B_i})$;
			\item\label{item:loop} $\sem{\repcom}(\qstate)= \bigvee_k \sem{S^k}(\qstate)$, where
			$S\define \repcom$,
			$S^0 \define \abort$, and for any $k\geq 0$,
			\[S^{k+1} \define \iif\ \square_{i=1}^n B_i\ra S_i; S^k\ \square\ B_0 \ra \sskip\ \fii.\]
			Here $B_0\define \bigwedge_{i=1}^n \neg B_i$. 
			Thus $
			\sem{S}(\qstate) = \qstate |_{B_0} + \sum_{i=1}^{n}\sem{S}(\sem{S_i}(\qstate |_{B_i})).
			$
		\end{enumerate}
	\end{lemma}
	\begin{proof}
		Similar to that of~\cite[Lemma 4.6]{feng2020quantum}.
	\end{proof}
	\section{Omitted proofs}
	
	\begin{proof}[Proof of Lemma~\ref{lem:supstep}]
		Induction on the structure of $S$.
	\end{proof}
	
	\begin{proof}[Proof of Lemma~\ref{lem:infcom}]
		This can be easily seen from the fact that the only successor configuration of a terminal one under $\ra$ is itself.
	\end{proof}
	
	\subsection{Proof of Theorem~\ref{thm:determ}}
	To prove Theorem~\ref{thm:determ}, we first introduce some notions. Note that from Table~\ref{tbl:opsemantics}, any transition $\conf{S}{\cstate}{\rho}\ra \mu$ of a distributed program $S$ must be obtained by using (Paral) or (Comm). To make it clear which processes are involved in the transition, we write $\conf{S}{\cstate}{\rho}\rto{k} \mu$ if it is caused by a local action of process $S_k$. Similarly, we write $\conf{S}{\cstate}{\rho}\rto{(k,\ell)} \mu$ if it is caused by a communication  between processes $S_k$ and $S_\ell$ with $k<\ell$. Let $\t \define [n] \cup \{(k,\ell)\in [n]^2: k<\ell\}$, where $[n] \define \{1,\ldots, n\}$, be the set of possible transition labels.
	\begin{definition} Let $\pi = \{\mu_i : i\geq 0\}$ be a computation of $\<S, \sigma, \rho\>$. The (infinite) \emph{derivative tree} $T$ induced by $\pi$ is defined as follows: for all $i\geq 0$,
		\begin{enumerate}
			\item nodes at the $i$-th level of $T$ are support configurations of $\mu_i$. In particular, the root node of $T$ is $\<S, \cstate, \rho\>$;
			\item for any $i$-th level node $c$ (thus $c\in \supp{\mu_i}$) which is not a terminal, if $c\rto{A}  \nu$, $A\in \t$, is the transition from $c$ which contributes to the evolvement from $\mu_i$ to $\mu_{i+1}$, then there is an edge in $T$ from $c$ to each support configuration of $\nu$. Furthermore, these edges are labelled by action $A$ and their corresponding probabilities in $\nu$;
			\item for any terminal configuration $c$ at the $i$-th level, note that $c$ also appears at the $i+1$-th level. Then 
			there is an edge in $T$ from the $i$-th level $c$ to the $i+1$-th level $c$. Furthermore, this edge is labelled by a special symbol $\ast$ and probability 1.
		\end{enumerate}
	\end{definition}	
	Note that from a derivative tree $T$, we can easily recover the computation $\{\mu_i : i\geq 0\}$ as follows: for each $i\geq 0$, let $N_i$ be the set of nodes at the $i$-th level of $T$. Then
	\[
	\mu_i = \sum_{c\in N_i} p_c\cdot c
	\]
	where $p_c$ is the product of all the probabilities along the path from the root to $c$.
	
	\begin{definition}
		Let $\pi$ be a computation of $\<S, \sigma, \rho\>$, and $T$ its derivative tree.
		\begin{enumerate}
			\item A \emph{run} $r = \{c_i : i\geq 0\}$ of $\pi$ is a path of $T$ starting from the root node (thus $c_0 = \<S, \sigma, \rho\>$).
			\item The \emph{history} of a run $r = \{c_i : i\geq 0\}$ is a sequence $\{(E_i, A_i) \in 2^{\t} \times (\t\cup \{\ast\}) : i\geq 0\}$ such that $E_i$ is the set of transition labels that are enabled in $c_i$, while $A_i$ is the label on the edge $(c_i, c_{i+1})$ in $T$. Note that $A_i \in E_i$ whenever $E_i\neq \emptyset$.
		\end{enumerate}
	\end{definition}
	
	Fix arbitrarily a linear order $\le$ over $\t$. For example, we may let $A\le B$ if (1) $A\in [n]$ and $B\in [n]^2$, or (2) $A<B$ when both $A$ and $B$ are in $[n]$, or (3) $i<j$ when $A = (i, k)$ and $B=(j,\ell)$. 
	
	\begin{definition}
		A run is \emph{good} if its history $\{(E_i, A_i) : i\geq 0\}$ satisfies the following condition:
		\[
		\forall i\geq 0: (E_i \neq \emptyset \ra A_i = \min E_i)
		\]
		where $\min E_i$ is the minimum element in $E_i$ according to the linear order $\le$.
		A computation $\pi$ is good if all of its runs are good.
	\end{definition}
	
	We are now ready to prove the main theorem of this section, which says that all computations from a given input computes the same cq-state.
	
	\begin{proof}[Proof of Theorem~\ref{thm:determ}]
		Note that from any configuration $\<S, \sigma, \rho\>$, there exists a unique good computation. The main idea of the proof is that we can always transform the derivative tree $T$ of any computation into that of the good one starting from the same configuration, using some `commutativity' properties of transitions from different processes. Furthermore, this transformation does not change the computed cq-state. 
		
		Let $\pi = \{\mu_i : i\geq 0\}$ be a computation with $\mu_0 = \<S, \sigma, \rho\>\rto{A} \mu_1$, and $T$ its derivative tree. Suppose the good computation from $\<S, \sigma, \rho\>$ would choose $B$, $B\neq A$, as the first action. We show in the following how to transform $\pi$ into another (not necessarily good)  computation $\pi'$ with the first action being $B$, and they compute the same cq-state. To simplify the presentation, we assume $B = k$ for some $k\in [n]$ (the case when $B \in [n]^2$ is similar). 
		
		First, we prove that the $B$-transition must appear along every terminating run of $\pi$. To see this, suppose on the contrary there is a successful run $r$ in which no $B$-transition is executed. Note that any transition which does not involve $k$ cannot change the value of variables in $\cVar(S_k)$, and since $S_k$ is deterministic, at most one of the actions in $\t$ which involve $k$ is enabled at any moment. Consequently, $B$ will be continuously enabled along $r$, which is a contradiction since the quantum program in the last configuration of $r$ must be $E$.
		
		Now for any terminating run $r = \{c_i : i\geq 0\}$ of $\pi$ (thus $c_0 = \<S, \sigma, \rho\>$) with history $\{(E_i, A_i) : i\geq 0\}$, let $c_{i_B}\define \<R_{i_B}, \sigma_{i_B}, \rho_{i_B}\>$ be the first configuration in which $B$ is executed; that is, $A_{i_B} = B$, and $A_i \neq B$ for all $i<i_B$. From transition rule (Paral) in Table~\ref{tbl:opsemantics} and Lemma~\ref{lem:supstep}, let 
		\begin{equation}\label{eq:SB}
			\<S, \sigma, \rho\>\rto{B} \sum_{j\in J} p_j \cdot \<R_j, f_j(\sigma), \e_j(\rho)/p_j\>
		\end{equation}
		where $R_j = S_1\|\ldots \|S_{k,j}\| \ldots\|S_n$ for some $S_{k,j}$, $f_j$ only depends on $\cVar(S_k)$ but does not change the variables outside $\change(S_k)$, $\e_j$ is a super-operator acting on $\h_{\qVar(S_k)}$, and $p_j = \tr(\e_j(\rho))$.
		Then from the fact that along the path $c_0, c_1, \ldots, c_{i_B}$, no action involving $k$ is performed, the transition that happens at  $c_{i_B}$ in the computation $\pi$ has the form
		\begin{equation}\label{eq:RiB}
			\<R_{i_B}, \sigma_{i_B}, \rho_{i_B}\>\rto{B} \sum_{j\in J} p_{i_B,j} \cdot \<R_{i_B,j}, f_j(\sigma_{i_B}), \e_j(\rho_{i_B})/p_{i_B,j}\>
		\end{equation}
		where $R_{i_B,j} = S_1^{i_B}\|\ldots \|S_{k,j}\| \ldots\|S_n^{i_B}$
		whenever $R_{i_B} = S_1^{i_B}\|\ldots \|S_k\| \ldots\|S_n^{i_B}$, and $p_{i_B,j} = \tr(\e_j(\rho_{i_B}))$.
		
		For any $j\in J$, we are going to construct from $T$ a derivative tree $T_j$ where the first execution of the $B$-transition along any terminating run of $T$ is replaced by the corresponding $j$-th child in the $B$-transition; that is, $c_{i_B}$ is replaced by $\<R_{i_B,j}, f_j(\sigma_{i_B}), \e_j(\rho_{i_B})/p_{i_B,j}\>$. To be more specific, $T_j$ is constructed as follows.
		\begin{enumerate}
			\item Let the root of $T_j$ be $\<R_j, f_j(\sigma), \e_j(\rho)/p_j\>$. 
			\item To unfold $T_j$ from the root, we follow precisely the transitions taken by $T$ along each run $r$ until the configuration $c_{i_B}$ is reached. For such a finite path $c_0, c_1, \ldots, c_{i_B}$ in $T$, it is easy to see that the corresponding path in $T_j$ is $c_0', c_1', \ldots, c_{i_B}'$, where 
			\[
			c_i' \define \<S_1^i\|\ldots \|S_{k,j}\| \ldots\|S_n^i, f_j(\sigma_i), \e_j(\rho_i)/\tr(\e_j(\rho_i))\>
			\]
			whenever 
			\[
			c_i = \<S_1^i\|\ldots \|S_k\| \ldots\|S_n^i, \sigma_i, \rho_i\>.
			\]
			Here in each $c_i$ the $k$-th process must be $S_k$ since along the path $c_0, c_1, \ldots, c_{i_B}$ in $T$, no $B$-transition is executed. In particular, $c_{i_B}'$
			is precisely the $j$-th support configuration of the right-hand side distribution in Eq.~\eqref{eq:RiB}. Furthermore, it is easy to check that each pair of the corresponding edges in $T$ and $T_j$ along each run up to the respective $c_{i_B}$ are labelled with the same probability. 
			\item The subtree of $T_j$ rooted at $c_{i_B}'$ is the same as the subtree of $T$ rooted at $c_{i_B}'$ (from the above clause, $c_{i_B}'$ indeed appears in $T$ as a child node of $c_{i_B}$). 
		\end{enumerate}

		Finally, let $T'$ be a derivative tree where the root is $\<S, \sigma, \rho\>$, the action executed by the root is given in Eq.~\eqref{eq:SB}, and for each $j\in J$,  $T_j$ is the subtree starting from $\<R_j, f_j(\sigma), \e_j(\rho)/p_j\>$. Note also that  the above procedure transforms non-terminating runs to non-terminating runs. Thus obviously, the induced computation $\pi' = \{\mu_i' : i\geq 0\}$ computes the same cq-state as $\pi$.
		
		Repeat the above procedure, we will eventually transform any computation to the good one without changing the cq-state computed. That concludes the proof of the theorem.
	\end{proof}
	
	\subsection{Proof of Lemma~\ref{lem:welldef}}
	Clause (2) is easy. For (1), let $\qstate = \<\cstate, \rho\>$ with $\tr(\rho)=1$, and 
	$\pi \define \{\mu_i: i\geq 0\}$ a computation of $S$ starting in $\qstate$. We prove by induction on $i$ that $\qstate_{\mu_i}$ has countable support and
	$\tr(\qstate_{\mu_i}) \leq \tr(\rho)$.
	Thus the result holds for simple cq-states. The general case follows easily.

	\subsection{Proof of Theorem~\ref{thm:sequentialisation}}
	
	We first show a close relationship between the \emph{good} transitions of $S$ and the transitions of $T(S)$.
	\begin{lemma}\label{lem:transsts}
		For any configuration $\<S, \sigma, \rho\>$ where $S$ is a distributed quantum program,
		\begin{enumerate}
			\item  if the transition
			$
			\<S, \sigma, \rho\> \ra \sum_{i\in I} p_i\cdot \<S_i, \sigma_i, \rho_i\>
			$
			appears in the derivative tree of a good computation, then 
			$
			\<T(S), \sigma, \rho\> \ra \sum_{i\in I} p_i\cdot  \<T(S_i), \sigma_i, \rho_i\>
			$
			is the (unique) transition from $\<T(S), \sigma, \rho\>$;
			\item conversely, if 
			$
			\<T(S), \sigma, \rho\> \ra \sum_{i\in I}p_i\cdot  \<S_i', \sigma_i, \rho_i\>
			$
			then either $
			\<S, \sigma, \rho\> \ra \sum_{i\in I} p_i\cdot  \<S_i, \sigma_i, \rho_i\>
			$ appears in the derivative tree of a good computation and $S_i' = T(S_i)$,
			or $\<S, \sigma, \rho\>$ is a deadlock. In the latter case, $\sigma \models \mathit{BLOCK}\wedge \neg \mathit{TERM}$.
		\end{enumerate}
	\end{lemma}
	\begin{proof}
		Easy from the definitions of $T(S)$, which is a deterministic quantum program, and the good computation of $S$. Furthermore, if $\<S, \sigma, \rho\>$ is a deadlock, then the classical state $\sigma$ must satisfy $\mathit{BLOCK}\wedge \neg \mathit{TERM}$.
	\end{proof}
	With this lemma, Theorem~\ref{thm:sequentialisation} can be  proved as follows.
	\begin{proof}[Proof of Theorem~\ref{thm:sequentialisation}]
		We need only prove the theorem for the case when $\qstate = \<\cstate, \rho\>$ with $\tr(\rho) = 1$. Let $\pi \define \{\mu_i : i\geq 0\}$ and $\pi' \define \{\mu_i' : i\geq 0\}$ be the computation of $T(S)$ and the good computation of $S$, both starting in $\<\cstate, \rho\>$, respectively. We are going to show that for any $i\geq 0$, $\qstate_{\mu_i} = \qstate_{\mu_i'}|_{\mathit{TERM}}$. Then the theorem follows by taking the least upper bounds of both sides.
		
		From Lemma~\ref{lem:transsts}, the derivative tree of $\pi$ has the same structure (including the probability weights along the edges) with that of $\pi'$, except for deadlock configurations. 
		However, Lemma~\ref{lem:transsts} also says that classical states in these deadlock configurations must satisfy $\mathit{BLOCK}\wedge \neg \mathit{TERM}$, and thus they will be excluded in computing $\qstate_{\mu_i'}|_{\mathit{TERM}}$. Note further that $\mathit{TERM}$ is satisfied by all the successfully terminating configurations in $\mu_i$; that is, $\qstate_{\mu_i} = \qstate_{\mu_i} |_{\mathit{TERM}}$. Thus $\qstate_{\mu_i} = \qstate_{\mu_i'}|_{\mathit{TERM}}$ as desired.
	\end{proof}
	
	\subsection{Proof of Theorems~\ref{thm:psc} and~\ref{thm:total}}

	The basic idea of proving the soundness and completeness of our proof systems is to employ weakest (liberal) preconditions. To this end, we extend the weakest (liberal) precondition semantics presented in~\cite{feng2020quantum} to sequential programs defined in Sec.~\ref{sec:seq}.
	Note that we do not have to extend it further to distributed programs, thanks to the sequentialisation theorem (Theorem~\ref{thm:sequentialisation}). Let $\qassertsh{\supseteq \qv(S)} \define\bigcup_{V\supseteq qv(S)} \qassertsh{V}$.

	\begin{definition}\label{def:weakest}
		Let $S$ be a sequential quantum program. The \emph{weakest precondition semantics} $wp.S$ and \emph{weakest liberal precondition semantics} $wlp.S$ of $S$ are both mappings 
		$$\qassertsh{\supseteq \qv(S)}\ra  \qassertsh{\supseteq \qv(S)}$$
		defined inductively in Table~\ref{tbl:wpsemantics}. To simplify notation, we use $xp$ to denote both $wp$ and $wlp$ whenever it is applicable for both of them.
	\end{definition}

	The following lemma shows a duality relation between the denotational and weakest (liberal) precondition semantics of sequential programs, which extends~\cite[Lemma 4.14]{feng2020quantum}.
	
	{\renewcommand{\arraystretch}{1.8}
		\begin{table}[t]
			\begin{lrbox}{\tablebox}
				\centering
				\begin{tabular}{l}
					\begin{tabular}{ll}
						$xp.\sskip.\qassert = \qassert$ & $wlp.\abort.\qassert = \top_{V}$ \hspace{3em} $wp.\abort.\qassert = \bot_{V}$\\
						$xp.(x:=e).\qassert =  \qassert[\subs{x}{e}]$ 
						&$xp.(x\rassign g).\qassert =  \sum_{d\in D_{\mathit{type(x)}}}g(d)\cdot \qassert[\subs{x}{d}]$\\
						$xp.(\bar{q}\apply U).\qassert  = U_{\bar{q}}^\dag \qassert U_{\bar{q}}$ &$xp.(q:=0).\qassert =  \sum_{i=0}^{d_{q}-1} \qiz \qassert\qzi$\\
						$xp.(S_0; S_1).\qassert = xp.S_0.(xp.S_1.\qassert)$ &	$xp.(x:= \mymeas\ \m[\bar{q}]).\qassert  = \sum_{i\in I}M_i^\dag \qassert[\subs{x}{i}]M_i$
						\\
					\end{tabular}\\ 
					\begin{tabular}{c}	
						$wlp.(\altercom).\qassert = \sum_{i=1}^{n} B_i\wedge wlp.S_i.\qassert + \bigwedge_{i=1}^n \neg B_i$\\				
						$wp.(\altercom).\qassert = \sum_{i=1}^{n} B_i\wedge wp.S_i.\qassert$\\
						$wlp.(\repcom).\qassert = \bigwedge_{k\geq 0} \qassert_k$, where 
						$\qassert_0 \define \top_{V}$, and for any $k\geq 0$,\\
						$\qassert_{k+1} \define \sum_{i=1}^{n} B_i\wedge wlp.S_i.\qassert_k + \bigwedge_{i=1}^n \neg B_i \wedge \qassert.$\\
						$wp.(\repcom).\qassert = \bigvee_{k\geq 0} \qassert_k$, where 
						$\qassert_0 \define \bot_{V}$, and for any $k\geq 0$,\\
						$\qassert_{k+1} \define \sum_{i=1}^{n} B_i\wedge wp.S_i.\qassert_k + \bigwedge_{i=1}^n \neg B_i \wedge \qassert.$
					\end{tabular}			
				\end{tabular}
			\end{lrbox}
			\resizebox{\textwidth}{!}{\usebox{\tablebox}}\\
			\vspace{4mm}
			\caption{Weakest (liberal) precondition semantics for sequential programs, where $xp\in \{wp, wlp\}$ and $V = \qv(\qassert)$. 
			}
			\label{tbl:wpsemantics}
		\end{table}
	}
	
	\begin{lemma}\label{lem:wpwlp}
		Let $S$ be a sequential quantum program, $\qstate$ a cq-state, and $\qassert$ a cq-assertion with $qv(\qstate) \supseteq qv(\qassert)\supseteq qv(S)$. Then
		\begin{enumerate}
			\item $qv(wp.S.\qassert) = qv(wlp.S.\qassert) = qv(\qassert)$;
			\item $\Exp(\qstate\models wp.S.\qassert) =  \Exp(\sem{S}(\qstate)\models \qassert)$;
			\item $\Exp(\qstate\models wlp.S.\qassert) =  \Exp(\sem{S}(\qstate)\models \qassert) + \tr(\qstate) -  \tr(\sem{S}(\qstate))$.
		\end{enumerate}
	\end{lemma}
	\begin{proof}
		We prove this lemma by induction on the structure of $S$. The basis cases are easy from the definition. We only show the following cases for clause (3) as examples. Let $V \define qv(\qassert)$.
		\begin{itemize}
			\item Let $S\define \altercom$. Note that $B_i$'s are mutually exclusive. Then 
			\begin{align*}
				\Exp(\qstate\models wlp.S.\qassert) &=\Exp\left(\qstate\models \sum_{i=1}^{n} B_i\wedge wlp.S_i.\qassert + \bigwedge_{i=1}^n \neg B_i\right)\\
				& = 	\sum_{i=1}^{n} \Exp(\qstate|_{B_i}\models wlp.S_i.\qassert) + \Exp\left(\qstate\models \top_{V} - \sum_{i=1}^{n} B_i\right) \\
				& =	\sum_{i=1}^{n}\left[\Exp\left( \sem{S_i}(\qstate|_{B_i})\models \qassert\right) +\tr(\qstate|_{B_i}) - \tr(\sem{S_i}(\qstate|_{B_i})) \right]\\
				& \hspace{2em}+  \tr(\qstate) - \sum_{i=1}^{n} \tr(\qstate|_{B_i})\\
				& = \Exp(\sem{S}(\qstate)\models \qassert) +  \tr(\qstate) - \tr(\sem{S}(\qstate)).
			\end{align*}
			Here the second equality follows from Lemma~\ref{lem:bpdeg}, the third one from the inductive hypothesis, and the last one from Lemma~\ref{lem:iout}.
			\item Let $S\define \repcom$ and $\qassert_k$, $k\geq 0$, be defined as in Table~\ref{tbl:wpsemantics} for the $wlp$ semantics of $\repcom$.
			First, we show by induction that for any $k\geq 0$ and $\qstate'\in \qstatesh{V}$,
			\begin{equation}\label{eq:tmp4.9.254}
				\Exp(\qstate' \models \qassert_k) = \Exp(\sem{S^k}(\qstate') \models \qassert) + \tr(\qstate') - \tr(\sem{S^k}(\qstate'))
			\end{equation}
			where $S^k$ is defined as in Lemma~\ref{lem:iout}.
			The case of $k=0$ follows from the definition. Let $B\define\bigwedge_{i=1}^n \neg B_i$. We further calculate from Lemmas~\ref{lem:bpdeg} and~\ref{lem:iout} that
			\begin{align*}
				& \Exp(\qstate' \models \qassert_{k+1})\\
				&= \Exp(\qstate'\models \sum_{i=1}^{n} B_i\wedge wlp.S_i.\qassert_k) + \Exp(\qstate' \models B\wedge \qassert)\\
				&=\sum_{i=1}^{n} \Exp(\qstate'|_{ B_i} \models  wlp.S_i.\qassert_k) + \Exp(\qstate'|_{B} \models \qassert)\\
				& =	\sum_{i=1}^{n}\left[\Exp\left( \sem{S_i}(\qstate'|_{B_i})\models \qassert_k\right) +\tr(\qstate'|_{B_i}) - \tr(\sem{S_i}(\qstate'|_{B_i})) \right]+ \Exp(\qstate'|_{B} \models \qassert)\\
				&=
				\sum_{i=1}^{n}\left[\Exp(\sem{S^k}(\sem{S_i}(\qstate'|_{B_i})) \models \qassert) + \tr(\sem{S_i}(\qstate'|_{B_i})) - \tr(\sem{S^k}(\sem{S_i}(\qstate'|_{B_i})))\right]\\
				& \hspace{2em} + \sum_{i=1}^{n}\left[ \tr(\qstate'|_{B_i}) - \tr(\sem{S_i}(\qstate'|_{B_i})) \right] + \Exp(\qstate'|_{B} \models  \qassert) \\
				%	&=\Exp(\qstate'|_{\neg b} +\sem{\while^n}(\sem{S'}(\qstate'|_{b})) \models \qassert)\\
				&=\Exp(\sem{S^{k+1}}(\qstate') \models \qassert)  + \tr(\qstate') - \tr(\sem{S^{k+1}}(\qstate')),
			\end{align*}
			where the fourth equality follows from the induction hypothesis, and 
			the last one from the fact that  
			$
			\sem{S^{k+1}}(\qstate') = \sum_{i=1}^{n} \sem{S^k}(\sem{S_i}(\qstate'|_{B_i})) + \qstate'|_{B} 
			$
			and $\tr(\qstate') = \sum_{i=1}^{n} \tr(\qstate'|_{B_i}) + \tr(\qstate'|_{B})$.
			
			With Eq.~\eqref{eq:tmp4.9.254}, we have from Lemma~\ref{lem:qasset} that
			$$\Exp(\qstate \models  wlp.S.\qassert) =\Exp(\qstate \models  \bigwedge_{k\geq 0} \qassert_k) =  \Exp(\sem{S}(\qstate) \models \qassert) + \tr(\qstate) - \tr(\sem{S}(\qstate)).$$
		\end{itemize}
	\end{proof}
	
	The following two lemmas, which extend Lemmas 4.16 and 4.17 in~\cite{feng2020quantum}, respectively, can be similarly shown for our sequential programs. The proofs are omitted here.
	\begin{lemma}\label{lem:wpcorres}
		Let $S$ be a sequential program, $\qstate$ a cq-state, and $\qassert$ a cq-assertion with $qv(\qstate) \supseteq qv(\qassert)\supseteq qv(S)$. Let $xp\in \{wp, wlp\}$. Then
		\begin{enumerate}
			\item\label{cl:wpwlp} $wp.S.\qassert + wlp.S.(\top_{qv(\qassert)}-\qassert) = \top_{qv(\qassert)}$;
			\item\label{cl:wpmono}  the function $xp.S$ is monotonic; that is, for all $\qassert_1 \le \qassert_2$,
			$
			xp.S.\qassert_1 \le xp.S.\qassert_2;
			$
			\item\label{cl:wplinear}  the function $wp.S$ is linear; that is, for all $\qassert_1, \qassert_2\in \qasserts$, 
			\[
			wp.S.(\lambda_1\qassert_1 + \lambda_2\qassert_2) = 
			\lambda_1 wp.S.\qassert_1 + \lambda_2 wp.S.\qassert_2;\]
			\item\label{cl:wlplinear}  the function $wlp.S$ is affine-linear; that is, for all $\qassert_1, \qassert_2\in \qasserts$ and $\lambda_1 + \lambda_2 =1$,
			\[
			wlp.S.(\lambda_1\qassert_1 + \lambda_2\qassert_2) = 
			\lambda_1 wlp.S.\qassert_1 + \lambda_2 wlp.S.\qassert_2.\]
			
			\item\label{cl:superoper} if $W\cap qv(\qassert) \subseteq V \subseteq qv(\qassert)$, $(V\cup W)\cap qv(S) = \emptyset$, and $\f_{V\ra W}$ is a completely positive and sub-unital super-operator, then
			$$
			\f_{V\ra W}(wp.S.\qassert)  = wp.S.(\f_{V\ra W}(\qassert))
			$$
			and 
			$$\f_{V\ra W}(wlp.S.\qassert)  \le wlp.S.(\f_{V\ra W}(\qassert)).
			$$
			The equality holds for $wlp$ as well if $\f_{V\ra W}$ is unital;
		\end{enumerate}
	\end{lemma}
	
	\begin{lemma}\label{lem:lesssimpre}
		Let $S$ be a sequential program, and $\qassert$ and $\qassertp$ are cq-assertions. Then 
		\begin{align*}
			\models_{{\mathit{\tot}}} \ass{\qassert}{S}{\qassertp} \quad &\mbox{ iff }\quad
			\qassert \lesssim wp.S.(\qassertp\otimes I_{qv(S)\backslash \qv(\qassertp)})\\
			\models_{\mathit{par}} \ass{\qassert}{S}{\qassertp} \quad &\mbox{ iff }\quad
			\qassert \lesssim wlp.S.(\qassertp\otimes I_{qv(S)\backslash \qv(\qassertp)}).
		\end{align*}
		In particular, if $qv(\qassert) = qv(\qassertp) \supseteq qv(S)$, then 
		\begin{align*}
			\models_{{\mathit{\tot}}} \ass{\qassert}{S}{\qassertp} \quad &\mbox{ iff }\quad \qassert \le wp.S.\qassertp\\
			\models_{\mathit{par}} \ass{\qassert}{S}{\qassertp} \quad &\mbox{ iff }\quad
			\qassert \le wlp.S.\qassertp.
		\end{align*}
	\end{lemma}

	The next lemma shows a closed relationship between the correctness of a distributed quantum program $S$ and its sequentialisation $T(S)$.
	
	\begin{lemma}\label{lem:reldistseq}
		For any distributed program $S$ and a cq-assertions $\qassert$ and $\qassertp$,
		\begin{align*}
			\models_{{\mathit{\tot}}} \ass{\qassert}{S}{\qassertp\wedge \mathit{TERM}}\ \   &\mbox{ iff }\ \  \models_{{\mathit{\tot}}} \ass{\qassert}{T(S)}{\qassertp\wedge \mathit{TERM}}\\
			\models_{\mathit{par}} \ass{\qassert}{S}{\qassertp\wedge \mathit{TERM} } \ \  &\mbox{ iff }\ \ 
			\models_{\mathit{par}} \ass{\qassert}{T(S)}{\qassertp\wedge \mathit{TERM} + \neg \mathit{TERM} \wedge \mathit{BLOCK}}.
		\end{align*}
	\end{lemma}
	\begin{proof}
		The first equivalence is direct from Theorem~\ref{thm:sequentialisation}. For the second one,
		let $\qassertp' \define \qassertp\wedge \mathit{TERM} + \neg \mathit{TERM} \wedge \mathit{BLOCK}$. It suffices to prove for any $\qstate\in \qstatesh{V}$ with $V\supseteq qv(\qassertp, \qassert, S)$, 
		\begin{equation}\label{eq:tmp4.8.203}
			\Exp(\sem{T(S)}(\qstate)\models \qassertp')-  \tr(\sem{T(S)}(\qstate)) = \Exp(\sem{S}(\qstate)\models \qassertp \wedge \mathit{TERM}) - \tr(\sem{S}(\qstate)) 
		\end{equation}
		Note that all support configurations in $\sem{T(S)}(\qstate)$ satisfy $\mathit{BLOCK}$. Thus \[\sem{T(S)}(\qstate) = \sem{S}(\qstate) + \sem{T(S)}(\qstate)|_{\neg \mathit{TERM}\wedge \mathit{BLOCK}}
		\] from Theorem~\ref{thm:sequentialisation}, and   
		\[\tr(\sem{T(S)}(\qstate)) = \tr(\sem{S}(\qstate)) + \Exp\left(\sem{T(S)}(\qstate) \models \neg \mathit{TERM}\wedge \mathit{BLOCK}\right).
		\]
		Then Eq.~\eqref{eq:tmp4.8.203} follows easily from the first equivalence.
	\end{proof}
	
	We are now ready to prove the soundness and completeness of our proof systems.
	
	\begin{proof}[Proof of Theorem~\ref{thm:psc}]
		\textbf{Soundness}: We need only to show that each rule in Table~\ref{tbl:psystem} is valid in the sense of partial correctness. The proof is divided into two steps:
		\begin{enumerate}
			\item We first prove by structural induction that the proof rules are sound for sequential programs (thus the rule (Dist) is no applicable).
			We take (Rep) as an example; the others are simpler. 
			Let $S\define \repcom$, and
			$\models_{\pal}  \ass{B_i\wedge \qassert}{S_i}{\qassert}$ for all $1\leq i\leq n$. Without loss of generality, we assume $qv(S) \subseteq qv(\qassert)$. Then $B_i\wedge \qassert\le wlp.S_i.\qassert$ from Lemma~\ref{lem:lesssimpre}. We now prove by induction on $k$ that
			$ \qassert \le \qassert_k$
			for any $k\geq 0$, where $\qassert_k$ is defined as in Table~\ref{tbl:wpsemantics} for the $wlp$ semantics of $\repcom$ when the postcondition is $\bigwedge_{i=1}^n\neg B_i\wedge \qassert$. The case when $k=0$ is trivial. Then we calculate
			\begin{align*}
				\qassert_{k+1} &= \sum_{i=1}^{n} B_i\wedge wlp.S_i.\qassert_k + \bigwedge_{i=1}^n \neg B_i \wedge \qassert\\
				&\ge\sum_{i=1}^{n} B_i\wedge wlp.S_i.\qassert + \bigwedge_{i=1}^n \neg B_i \wedge \qassert\\
				&\ge \sum_{i=1}^{n} B_i\wedge \qassert + \bigwedge_{i=1}^n \neg B_i \wedge  \qassert=\qassert,
			\end{align*}
			where the first inequality follows from the induction hypothesis and Lemma~\ref{lem:wpcorres}.
			Thus \[
			\qassert \le wlp.(\repcom).\left(\bigwedge_{i=1}^n \neg B_i\wedge \qassert\right),
			\] and so 
			\[\models_{\mathit{par}} \ass{\qassert}{\repcom}{\bigwedge_{i=1}^n \neg B_i\wedge \qassert}
			\] by Lemma~\ref{lem:lesssimpre}.
			\item For generic distributed program, the only relevant rules are (Imp) and (Dist). The former is direct from Lemma~\ref{lem:lesssimpre}. For (Dist), let $S\define S_1\| \ldots \| S_n$ and $T(S)$ be its sequentialisation defined in Sec.~\ref{sec:sequentialisation}.
			Suppose $\models_{\pal} \ass{\qassert}{S_{1,0}; \ldots;  S_{n,0}}{\qassertp}$, and for all $(i,j,k,\ell)\in \Gamma$,
			\[\models_{\pal} \ass{B_{i,j}\wedge B_{k,\ell}\wedge\qassertp}{\mathit{Effect}(\alpha_{i,j}, \alpha_{k,\ell}); S_{i,j}; S_{k,\ell}}{\qassertp}.
			\]
			Note that $T(S)$ is sequential.
			First, by the soundness of (Imp) for sequential programs, we have
			\[\models_{\pal}\ass{B_{i,j}\wedge B_{k,\ell}\wedge B_i\wedge\qassertp}{\mathit{Effect}(\alpha_{i,j}, \alpha_{k,\ell}); S_{i,j}; S_{k,\ell}}{\qassertp}.
			\]
			Then $\models_{\pal} \ass{\qassert}{T(S)}{\qassertp \wedge \mathit{BLOCK}}$
			by using the soundness of (Seq) and (Rep) for sequential programs. Note that $\qassertp \wedge \mathit{BLOCK} \le \qassertp'$ where $\qassertp'$ is defined in Lemma~\ref{lem:reldistseq}. Thus from (Imp) and Lemma~\ref{lem:reldistseq} we have $\models_{\pal} \ass{\qassert}{S}{\qassertp \wedge \mathit{TERM}}$.
		\end{enumerate}
		
		\textbf{Completeness}: The proof for completeness is also divided into two steps:
		\begin{enumerate}
			\item We first prove by induction on the structure of $S$ that for any $\qassert$ and sequential program $S$ with $qv(S)\subseteq qv(\qassert)$,
			$\vdash_{\mathit{par}} \ass{wlp.S.\qassert}{S}{\qassert}.
			$
			We take the case for loops as an example. 
			Let $S \define \repcom$ and $\qassertp \define wlp.S.\qassert$.
			By induction, we have 
			$\vdash_{\mathit{par}} \ass{wlp.S_i.\qassertp}{S_i}{\qassertp}$ for any $1\leq i\leq n$.
			Note that
			\[\qassertp = \sum_{i=1}^{n} B_i\wedge wlp.S_i.\qassertp + \bigwedge_{i=1}^n \neg B_i \wedge \qassert.
			\]
			Thus $B_i\wedge \qassertp = B_i\wedge wlp.S_i.\qassertp \le wlp.S_i.\qassertp$ and 
			so $\vdash_{\mathit{par}} \ass{B_i\wedge \qassertp}{S_i}{\qassertp}$ by the (Imp) rule. Now using (Rep) we have
			$\vdash_{\mathit{par}} \ass{\qassertp}{\repcom}{\bigwedge_{i=1}^n \neg B_i\wedge \qassertp}$
			and the result follows from the fact that $\bigwedge_{i=1}^n \neg B_i\wedge \qassertp =\bigwedge_{i=1}^n \neg B_i\wedge \qassert \le \qassert$. 
			\item  Let $S\define S_1\| \ldots \| S_n$ and $T(S)$ its sequentialisation defined in Sec.~\ref{sec:sequentialisation}.
			Suppose $\models_{\pal} \ass{\qassert}{S}{\qassertp}$. Note that for any $\qstate\in \qstatesh{V}$ with $V\supseteq qv(\qassertp, \qassert, S)$, all support configurations in $\sem{S}(\qstate)$ satisfy $\mathit{TERM}$. Thus 	
			$\models_{\pal} \ass{\qassert}{S}{\qassertp\wedge \mathit{TERM}}$. Then from Lemma~\ref{lem:reldistseq}, we have 
			$\models_{\pal} \ass{\qassert}{T(S)}{\qassert'}$ where $\qassert' \define \qassertp\wedge \mathit{TERM} + \neg \mathit{TERM} \wedge \mathit{BLOCK}$, and thus $\qassert \lesssim wlp.T(S).\qassert'$.

			Let $\qassertp' \define wlp.\ddo.\qassert'$ where $\ddo$ is the $\ddo$-loop in $T(S)$.
			As $T(S)$ is sequential, we have from the above clause that
			\[
			\vdash_{\pal} \ass{\qassertp'}{\ddo}{\qassert'} \hspace{2em} \mbox{and} \hspace{2em}
			\vdash_{\pal} \ass{wlp.S_0.\qassertp'}{S_0}{\qassertp'}
			\]
			where $S_0 \define S_{1,0}; \ldots;  S_{n,0}$. Note that
			\[\qassertp' = \sum_{(i,j,k,\ell)\in \Gamma} B_{i,j}\wedge B_{k,\ell}\wedge B_i \wedge wlp.S_{i,j}^{k,\ell}.\qassertp' + \mathit{BLOCK} \wedge \qassert'
			\]
			where $S_{i,j}^{k,\ell}\define \mathit{Effect}(\alpha_{i,j}, \alpha_{k,\ell}); S_{i,j}; S_{k,\ell}$. Thus
			\[
			B_{i,j}\wedge B_{k,\ell}\wedge\qassertp' = B_{i,j}\wedge B_{k,\ell}\wedge B_i \wedge wlp.S_{i,j}^{k,\ell}.\qassertp' \le wlp.S_{i,j}^{k,\ell}.\qassertp',
			\]
			and so 
			\[\vdash_{\pal} \ass{B_{i,j}\wedge B_{k,\ell}\wedge\qassertp'}{S_{i,j}^{k,\ell}}{\qassertp'}.
			\]
			by (Imp) and the completeness result for sequential quantum programs. Note that $wlp.S_0.\qassertp' = wlp.T(S).\qassert'$. Applying (Dist) and (Imp), we derive
			$\vdash_{\pal} \ass{\qassert}{S}{\qassertp'\wedge \mathit{TERM}}$,
			and the result follows from the fact that $\qassertp'\wedge \mathit{TERM} =  \qassert' \wedge  \mathit{TERM}= \qassertp\wedge \mathit{TERM} \le \qassertp$.
		\end{enumerate}
	\end{proof}
	
	The proof for total correctness is more involved.
	\begin{proof}[Proof of Theorem~\ref{thm:total}]
		\textbf{Soundness}:  Similar to the partial correctness case, the proof is divided into two steps:
		\begin{enumerate}
			\item We first prove by structural induction that the proof rules in Table~\ref{tbl:psystem} with the corresponding rules replaced by those in Table~\ref{tbl:tsystem} are sound for sequential programs (thus the rule (Dist-T) is no applicable),  in the sense of total correctness. Again,
			we take (Rep-T) as an example. Let $S\define \repcom$,
			$\models_{\tot}  \ass{B_i\wedge \qassert}{S_i}{\qassert}$ for all $1\leq i\leq n$, and 
			$\{\qassertp_k : k\geq 0\}$ be a sequence of
			$\qassert$-ranking assertions for $S$.
			Without loss of generality, we assume $qv(S) \subseteq qv(\qassert)$. We now prove by induction on $k$ that
			$ \qassert \le \qassert_k + \qassertp_k$
			for any $k\geq 0$, where $\qassert_k$ is defined as in Table~\ref{tbl:wpsemantics} for the $wp$ semantics of $\repcom$ when the postcondition is $\bigwedge_{i=1}^n\neg B_i\wedge \qassert$. The case when $k=0$ is from the fact that $\qassert \le \qassertp_0$. 
			Then from the inductive hypothesis and Lemmas~\ref{lem:wpcorres} and~\ref{lem:lesssimpre},
			\[
			B_i\wedge \qassert\le wp.S_i.\qassert \le wp.S_i.\qassert_k + wp.S_i.\qassertp_k,
			\]
			and so
			\begin{align*}
				\qassert_{k+1} + \qassertp_{k+1}& \ge \sum_{i=1}^{n} B_i\wedge wp.S_i.\qassert_k + \bigwedge_{i=1}^n \neg B_i \wedge \qassert + \sum_{i=1}^{n} B_i \wedge wp.S_i.\qassertp_{k}\\
				&\ge \sum_{i=1}^{n} B_i\wedge \qassert + \bigwedge_{i=1}^n \neg B_i \wedge  \qassert=\qassert,
			\end{align*}
			where the first inequality follows from the definition of ranking assertions and the fact that $B_i$'s are mutually exclusive, and the second one from the induction hypothesis.
			Thus \[
			\qassert \le wp.(\repcom).\left(\bigwedge_{i=1}^n \neg B_i\wedge \qassert\right)
			\] by noting that $\bigwedge_k \qassertp_k = \bot_V$, and so 
			\[\models_{\mathit{\tot}} \ass{\qassert}{\repcom}{\bigwedge_{i=1}^n \neg B_i\wedge \qassert}
			\] as desired.
			\item For generic distributed programs, again we only consider (Dist). The proof is similar to the case for partial correctness, by noting the following two facts: for any distributed program $S$ and cq-assertion $\qassertp$,
			\begin{itemize}
				\item ranking assertions for $S$ are also ranking assertions for the $\ddo$-loop of $T(S)$;
				\item from the assumption $\qassertp \wedge \mathit{BLOCK} \lesssim \mathit{TERM}$ we have
				$\qassertp \wedge \mathit{BLOCK} =\qassertp \wedge \mathit{TERM}$. 
			\end{itemize}

			%		Let $S\define S_1\| \ldots \| S_n$ and $T(S)$ be its sequentialisation defined in Sec.~\ref{sec:sequentialisation}.
			%		Suppose $\models_{\tot} \ass{\qassert}{S_{1,0}; \ldots;  S_{n,0}}{\qassertp}$, for all $(i,j,k,\ell)\in \Gamma$,
			%		\[\models_{\tot} \ass{B_{i,j}\wedge B_{k,\ell}\wedge\qassertp}{\mathit{Effect}(\alpha_{i,j}, \alpha_{k,\ell}); S_{i,j}; S_{k,\ell}}{\qassertp},
			%		\]
			%		$\qassertp \wedge \mathit{BLOCK} \lesssim \mathit{TERM}$,	
			%		and 
			%		$\{\qassertp_k : k\geq 0\}$ be a sequence of
			%		$\qassertp$-ranking assertions for $S$.
			%		
			%		First, by the soundness of (Imp) for sequential programs we have
			%		\[\models_{\tot}\ass{B_{i,j}\wedge B_{k,\ell}\wedge B_i\wedge\qassertp}{\mathit{Effect}(\alpha_{i,j}, \alpha_{k,\ell}); S_{i,j}; S_{k,\ell}}{\qassertp}.
			%		\]
			%		Note that $\{\qassertp_k : k\geq 0\}$ is actually also a sequence of
			%		$\qassertp$-ranking assertions for the $\ddo$-loop of $T(S)$.
			%		Then $\models_{\tot} \ass{\qassert}{T(S)}{\qassertp \wedge \mathit{BLOCK}}$
			%		by using the soundness of (Seq) and (Rep-T) for sequential programs. Furthermore, from the assumption $\qassertp \wedge \mathit{BLOCK} \lesssim \mathit{TERM}$ we have
			%		$\qassertp \wedge \mathit{BLOCK} =\qassertp \wedge \mathit{TERM}$. The result then follows from Lemma~\ref{lem:reldistseq}. 
		\end{enumerate}
		
		\textbf{Completeness}: The proof for completeness is also divided into two steps:
		\begin{enumerate}
			\item We first prove by induction on the structure of $S$ that for any $\qassert$ and sequential program $S$ with $qv(S)\subseteq qv(\qassert)$,
			$\vdash_{\mathit{\tot}} \ass{wp.S.\qassert}{S}{\qassert}.
			$
			Again, we take the case for loops as an example. 
			Let $S \define \repcom$ and $\qassertp \define wp.S.\qassert$.
			By induction, we have 
			$\vdash_{\mathit{\tot}} \ass{wp.S_i.\qassertp}{S_i}{\qassertp}$ for any $1\leq i\leq n$.
			Note that
			\[\qassertp = \sum_{i=1}^{n} B_i\wedge wp.S_i.\qassertp + \bigwedge_{i=1}^n \neg B_i \wedge \qassert.
			\]
			Thus $B_i\wedge \qassertp = B_i\wedge wp.S_i.\qassertp \le wp.S_i.\qassertp$ and 
			so $\vdash_{\mathit{\tot}} \ass{B_i\wedge \qassertp}{S_i}{\qassertp}$ by the (Imp) rule. 
			
			Let $\qassert_0 \define wp.S.\top_{qv(\qassert)}$ and 
			$\qassert_{k+1} \define \sum_{i=1}^{n} B_i\wedge wp.S_i.\qassert_k$. We are going to show that $\{\qassert_k : k\geq 0\}$ are $\qassertp$-ranking assertions for $S$. First, note that 
			$$\qassert_ 1 =  \sum_{i=1}^{n} B_i \wedge wp.S_i.\qassert_0 \le \bigwedge_{i=1}^n \neg B_i \wedge \top_{qv(\qassert)} +  \sum_{i=1}^{n} B_i\ \wedge wp.S_i.\qassert_0 = \qassert_0.$$
			So $\{\qassert_k : k\geq 0\}$ is decreasing by easy induction, using  Lemma~\ref{lem:wpcorres}\ref{cl:wpmono}. Next, as $\qassert \le \top_{qv(\qassert)}$, we have 
			$\qassertp\le \qassert_0$.  
			
			Finally, we prove that $\bigwedge_k \qassert_k =\bot_{qv(\qassert)}$.
			We show by induction on $k$ that for any $k\geq 0$ and $\qstate\in \qstatesh{qv(\qassert, S)}$,
			\begin{equation}\label{eq:induc}
				\Exp(\qstate\models \qassert_k) =
				\tr(\sem{S}(\qstate)) - \tr(\sem{S^k}(\qstate)).
			\end{equation}
			The case when $k=0$ is direct from Lemmas~\ref{lem:bpdeg} and \ref{lem:wpwlp}.
			We further calculate that
			\begin{align*}
				\Exp(\qstate\models \qassert_{k+1}) 
				&=\Exp\left(\qstate\models \sum_{i=1}^{n} B_i\wedge wp.S_i.\qassert_k\right)\\	
				&=\sum_{i=1}^{n} \Exp(\qstate|_{B_i} \models wp.S_i.\qassert_{k})\\
				&= \sum_{i=1}^{n} \Exp(\sem{S_i}(\qstate|_{B_i})\models \qassert_{k})\\
				&=  \sum_{i=1}^{n} \tr(\sem{S}(\sem{S_i}(\qstate|_{B_i}))) -  \sum_{i=1}^{n} \tr(\sem{S^k}(\sem{S_i}(\qstate|_{B_i})))\\
				&=\tr(\sem{S}(\qstate)) - \tr(\sem{S^{k+1}}(\qstate)).
			\end{align*}
			Here the second last equality is from induction hypothesis, and the last one from Lemma~\ref{lem:iout}.
			Note that the second term of the r.h.s of Eq.(\ref{eq:induc}) converges to the first one when $k$ goes to infinity. Thus
			$\lim_k \Exp(\qstate\models\qassert_k) = 0$, and so $\bigwedge_k \qassert_k = \emptydis_{qv(\qassert)}$ from the arbitrariness of $\qstate$.
			
			Now using (Rep-T) we have
			$\vdash_{\mathit{\tot}} \ass{\qassertp}{\repcom}{\bigwedge_{i=1}^n \neg B_i\wedge \qassertp}$
			and the result follows from the fact that $\bigwedge_{i=1}^n \neg B_i\wedge \qassertp =\bigwedge_{i=1}^n \neg B_i\wedge \qassert \le \qassert$. 
			
			\item The case for generic distributed programs $S$ is similar to that for partial correctness.
			The construction of ranking assertions for the $\ddo$-loop of $T(S)$, which also work for $S$, follows the same approach in the above clause.

		\end{enumerate}
	\end{proof}
	
	{\renewcommand{\arraystretch}{2.9}
		\begin{table}[t]
			\begin{lrbox}{\tablebox}
				\centering
				\begin{tabular}{l}
					\begin{tabular}{lc}
						(C-Rep-T)	& $\displaystyle\frac{
							\ass{B_i\wedge \qassert}{S_i}{\qassert},
							\ass{B_i\wedge\cassert \wedge t=z}{S_i}{t<z},
							\ \forall i\in \{1,\ldots, n\},\ \cassert \rightarrow t\geq 0	
						}{\ass{\qassert}{\repcom}{\qassert \wedge \bigwedge_{i=1}^n \neg B_i}}$ \\ 
						& where $\mathit{type}(z) = \mathit{type}(t) = \tyint$, $z\not\in \cVar(\cassert, B_i, t, S_i)$, $\qassert = \bigoplus_{i\in I} \<\cassert_i, M_i\>$ and $\cassert \define \bigvee_{i\in I} \cassert_i$.
						\\
						(C-Dist-T) & $\displaystyle\frac{
							{\renewcommand{\arraystretch}{1.5}
								\begin{tabular}{l}
									$\ass{\qassert}{S_{1,0}; \ldots;  S_{n,0}}{\qassertp},\ \cassert \rightarrow t\geq 0,\  \cassert\wedge \mathit{BLOCK} \rightarrow \mathit{TERM}$\\ 
									$\ass{B_{i,j}\wedge B_{k,\ell}\wedge\cassert \wedge t=z}{\mathit{Effect}(\alpha_{i,j}, \alpha_{k,\ell}); S_{i,j}; S_{k,\ell}}{t<z}, \forall (i,j,k,\ell)\in \Gamma$	\\
									$\ass{B_{i,j}\wedge B_{k,\ell}\wedge\qassertp}{\mathit{Effect}(\alpha_{i,j}, \alpha_{k,\ell}); S_{i,j}; S_{k,\ell}}{\qassertp}, \forall (i,j,k,\ell)\in \Gamma$	
								\end{tabular}
						}}	
						{\ass{\qassert}{S_1\|\ldots\|S_n}{\qassertp\wedge \mathit{TERM}}}$\\ 
						& 
						{\renewcommand{\arraystretch}{1.5}
							\begin{tabular}{ll}
								& where $\Gamma$, $\mathit{TERM}$, and $\mathit{BLOCK}$ are defined as in Sec.~\ref{sec:sequentialisation},\\ 
								& $\mathit{type}(z) = \mathit{type}(t) = \tyint$, $z\not\in \cVar(\cassert, t, S_1\|\ldots\|S_n)$, $\qassertp = \bigoplus_{i\in I} \<\cassert_i, M_i\>$, and $\cassert \define \bigvee_{i\in I} \cassert_i$.
							\end{tabular}		
						}
					\end{tabular}		
				\end{tabular}
			\end{lrbox}
			\resizebox{\textwidth}{!}{\usebox{\tablebox}}
			\vspace{1em}
			\caption{Auxiliary rules.}
			\label{tbl:auxrules}
		\end{table}
	}

	\section{Auxiliary Rules}
	We have provided sound and relatively complete proof systems for both partial and total correctness of distributed quantum programs. Thus in principle, these proof rules are sufficient for proving desired properties as long as they can be described faithfully with Hoare triple formulas.
	However, in practice, using these rules directly might be complicated. To simplify reasoning, we introduce two auxiliary proof rules in Table~\ref{tbl:auxrules} for the special case when a classical ranking function can be found to guarantee the (finite) termination of repetitive sequential (C-Rep-T) or distributed (C-Dist-T) quantum programs. More auxiliary proof rules (for deterministic quantum programs) can be found in~\cite{feng2020quantum,ying2019toward}.
	For the sake of convenience, we write $\<\cassert, |\psi\>\>$ for $\<\cassert, |\psi\>\<\psi|\>$, and $\cassert$ for $\cassert \wedge \top_V$ for some appropriate $V$.
	
	\begin{theorem} \label{thm:aux}
		The auxiliary rules presented in Table~\ref{tbl:auxrules} are sound with respect to total correctness. 
	\end{theorem}
	\begin{proof}%[Proof of Theorem~\ref{thm:aux}]
		First note that for any $i$, $\models_{\mathit{\tot}}\ass{B_i\wedge\cassert \wedge t=z}{S_i}{t<z}$ implies for any $\cstate \models B_i\wedge \cassert\wedge t=z$ and $\rho$,
		and any $\cstate'$ in the support of $\sem{S_i}(\cstate, \rho)$, we have $\cstate'\models t<z$. Then an argument similar to that for classical programs leads to the conclusion that all computations from $\<\repcom, \cstate, \rho\>$ terminates within $\cstate(t)$ steps, provided $\cstate \models \cassert$. That proves (C-Rep-T). The case for (C-Dist-T) is similar.
	\end{proof}
	\section{Case studies}
	
	To illustrate the effectiveness of the proof systems as well as the auxiliary rules presented in this paper, we employ them to verify the quantum teleportation protocol. A protocol to locally implement nonlocal gates is also investigated.
	
	\subsection{Verification of quantum teleportation}
	
	\begin{example}[Correctness of Quantum Teleportation]
		The correctness of quantum teleportation can be stated as follows: for any $|\psi\>\in \h_2$,
		\begin{equation}\label{eq:telcor}
			\vdash_{\tot}	\ass{|\psi\>_q\otimes |\beta\>_{q_1,q_2}}{\mathit{Teleport}}{|\psi\>_{q_2}}
		\end{equation}
		%Here $\qassert[q_2/q]$ denotes a quantum assertion in $\qassertsh{V'}$ with $V' \define V\backslash\{q\} \cup \{q_2\}$ which is identical to $\qassert$ except that the system $q$ on which it is applied is replaced by $q_2$. 

		%It is also worth pointing out that Eq.~\eqref{eq:telcor} says more than that the teleportation protocol faithfully transmits an arbitrary state $|\psi\>$ that originally resided in qubit $q$ to qubit $q_2$.
		%Let $\qassert = \sum_{m,n =0,1} \qassert_{m,n} \otimes |m\>_{q} \<n|$ for some $ \qassert_{m,n} \in \qassertsh{V\backslash\{q\}}$.
		
		The main technique of proving Eq.~\eqref{eq:telcor} is to employ rule (C-Dist-T). Let $t\define 2-stage_A$ and
		\begin{align*}
			\qassertp \define \frac{1}{4}\sum_{i,j =0,1}& \left(\cqs{stage_A = stage_B =0\wedge x_A = i \wedge z_A = j}{|j,i\>_{q,q_1} \otimes X^iZ^j |\psi\>_{q_2}} \right.\\
			+ & \cqs{stage_A =  stage_B =1 \wedge x_A = i \wedge z_A = j}{|j,i\>_{q,q_1} \otimes Z^j |\psi\>_{q_2}}\\
			+ & \left.\cqs{stage_A =  stage_B =2\wedge x_A = i \wedge z_A = j}{|j,i\>_{q,q_1} \otimes |\psi\>_{q_2}} \right).
		\end{align*}
		The proof consists of three parts.
		\begin{enumerate}
			\item We show that $\qassertp$ is a global invariant for the distributed programs $\mathit{Teleport}$. To this end, consider the first branch of the $\ddo$-loop in $T(\mathit{Telepor})$ presented in Example~\ref{ex:seqtel}:
			\begin{align*}		
				& \left\{stage_A =  stage_B = 0\wedge\qassertp\right\}\\
				&\left\{\frac{1}{4}\sum_{i,j=0,1} \cqs{x_A = i \wedge z_A = j}{|j,i\>_{q,q_1} \otimes X^iZ^j |\psi\>_{q_2}}\right\} & \mathit{(Imp)}\\
				& x_B := x_A;\\
				&\left\{\frac{1}{4}\sum_{i,j,k=0,1} \cqs{x_A = i \wedge z_A = j \wedge x_B=k}{|j,i\>_{q,q_1} \otimes X^kZ^j |\psi\>_{q_2}}\right\}& \mathit{(Assn)}\\
				& stage_A := 1;\\
				&\left\{\frac{1}{4}\sum_{i,j,k=0,1} \cqs{stage_A  =1 \wedge x_A = i \wedge z_A = j \wedge x_B=k}{|j,i\>_{q,q_1} \otimes X^kZ^j |\psi\>_{q_2}}\right\}& \mathit{(Assn)}\\
				& stage_B := 1;\\
				&\left\{\sum_{k=0,1} (x_B=k) \wedge \x_{q_2}^k(\qassertp)\right\}& \mathit{(Assn)}\\
				& \measstm{x_B=1}{q_2\apply X}{\sskip}\\
				&\left\{\qassertp\right\} & \mathit{(Alt)}
			\end{align*}
			where $\x$ is the Pauli-$X$ super-operator. Similarly, for the second branch, we can prove that
			\begin{align*}		
				& \left\{stage_A =  stage_B = 1\wedge\qassertp\right\}\\
				& z_B := z_A;\ stage_A := 2;\  stage_B := 2; \ \measstm{z_B=1}{q_2\apply Z}{\sskip} \\
				&\left\{\qassertp\right\}.
			\end{align*}
			
			\item We show that $t$ is a classical ranking function for the distributed program $\mathit{Teleport}$. Note that $\mathit{BLOCK} \equiv \bigwedge_{k=0,1} \neg \left(stage_A = stage_B = k\right)$, $$\mathit{TERM} \equiv \bigwedge_{k=0,1} \neg \left(stage_A = k\right) \wedge \bigwedge_{k=0,1} \neg \left(stage_B = k\right)$$
			and the classical part of $\qassertp$ is $\cassert \define \bigvee_{k=0}^2 \left(stage_A = stage_B = k\right)$.
			Then it is easy to check that
			$\cassert \rightarrow t\geq 0$ and
			$\cassert\wedge \mathit{BLOCK} \rightarrow \mathit{TERM}$.
			Furthermore, from
			\begin{align*}		
				& \left\{stage_A =  stage_B = 0\wedge\cassert \wedge 2-stage_A=z\right\}\\
				&\left\{1 < z\right\} & \mathit{(Imp)}\\
				& x_B := x_A;\\
				&\left\{1 < z\right\}& \mathit{(Assn)}\\
				& stage_A := 1;\\
				&\left\{2-stage_A < z\right\}& \mathit{(Assn)}\\
				& stage_B := 1; \ \measstm{x_B=1}{q_2\apply X}{\sskip}\\
				&\left\{2-stage_A < z\right\}& \mathit{(Assn,Alt)}
			\end{align*}
			and similarly for the second branch of the  $\ddo$-loop,
			the integer expression $t$ is indeed a classical ranking function for $\mathit{Teleport}$. 
			
			\item We show that the sequential part of $T(\mathit{Teleport})$ establishes $\qassertp$ from the precondition $|\psi\>_q\otimes |\beta\>_{q_1,q_2}$.	Let $|\psi\> = x|0\> + y|1\>$ for some $x,y\in \C$. Then
			\begin{align*}		
				& \left\{|\psi\>_q\otimes |\beta\>_{q_1,q_2}\right\}\\
				& q,q_1 \apply \textit{CNOT}; \\
				& \left\{\frac{1}{\sqrt{2}}\left( x |0\>_q (|00\> + 11\>)_{q_1, q_2} + y |1\>_q (|10\> + 01\>)_{q_1, q_2}\right)\right\}& \mathit{(Unit)}\\
				& q\apply H; \\
				& \left\{\frac{1}{2}\left( x (|0\> + |1\>)_q (|00\> + 11\>)_{q_1, q_2} + y (|0\> - |1\>)_q (|10\> + 01\>)_{q_1, q_2}\right)\right\}& \mathit{(Unit)}\\		
				%				& \left\{\frac{1}{2}\left( |00\>_{q,q_1} \otimes  (x|0\> + y|1\>)_{q_2} + \ldots +  11\>_{q,q_1} \otimes  (x|1\> - y|0\>)_{q_2}\right)\right\}& \mathit{(Imp)}\\
				&\left\{\sum_{i,j =0,1} \frac{1}{2}|j,i\>_{q,q_1} \otimes X^iZ^j |\psi\>_{q_2} \equiv \frac{1}{4}\sum_{i,j =0,1} \cqs{ \true}{ |j,i\>_{q,q_1} \otimes X^iZ^j |\psi\>_{q_2}}\right\} & \mathit{(Imp)}\\
				& z_A := \mymeas\ q;\  x_A := \mymeas\ q_1; \\
				&\left\{\frac{1}{4}\sum_{i,j =0,1} \cqs{ x_A = i \wedge z_A = j}{ |j,i\>_{q,q_1} \otimes X^iZ^j |\psi\>_{q_2}}\right\} & \mathit{(Meas)}\\
				& stage_A := 0;\ stage_B := 0;\\
				& \{\qassertp\} & \mathit{(Assn)}
			\end{align*}
			
		\end{enumerate}
		
		With the three parts shown above, we have from (C-Dist-T) that
		\[
		\vdash_{\tot}	\ass{|\psi\>_q\otimes |\beta\>_{q_1,q_2}}{\mathit{Teleport}}{\qassertp \wedge \mathit{TERM}}.
		\] 
		Then the desired result in Eq.~\eqref{eq:telcor} is obtained by noting that
		\[\qassertp \wedge \mathit{TERM}
		\equiv \frac{1}{4}\sum_{i,j =0,1} \cqs{stage_A =  stage_B =2\wedge x_A = i \wedge z_A = j}{|j,i\>_{q,q_1} \otimes |\psi\>_{q_2}},
		\]
		which is upper bounded above by $|\psi\>_{q_2}$ according to the order $\lesssim$.
	\end{example}
	
	\subsection{Local implementation of nonlocal quantum gates}
	In distributed quantum computing, one of the key tasks is to implement quantum gates between qubits that are located in different quantum computers. To illustrate the basic idea, we recall the protocol proposed in~\cite{eisert2000optimal} which implements a nonlocal CNOT gate  between two parties, say Alice and Bob, by employing only local quantum operations and classical communication, again with the help of a pre-shared entangled state.
	\begin{figure}[t]\centering
		\tikzset{
			my label/.append style={above right,xshift=0.3cm}
		}
		%		\begin{quantikz}[row sep=0.3cm,column sep=1cm]
		%			\lstick{$|\psi\>$} &\ctrl{1} & \qw &\qw &\qw & \gate{Z} &  \qw\\	  
		%			\lstick[3]{$|\beta\>$} &\targ{}  &   \meter{$x$} \vcw{2} & \\
		%			&\wave&&&&&&\\
		%			&\qw  &  \gate{X} & \ctrl{1} & \gate{H} &\meter{$z$} \vcw{-3}  \\
		%			\lstick{$|\phi\>$} & \qw & \qw & \targ{} & \qw &\qw &\qw
		%		\end{quantikz}
		\begin{quantikz}[row sep=0.3cm,column sep=1cm]
			\lstick{$|\psi\>$} &\ctrl{1} & \qw &\qw &\qw & \gate{Z} &  \qw\\	  
			\lstick[3]{$|\beta\>$} &\targ{}  &   \qw& \meter{$x$} & \cwbend{3} & \\
			&\wave&&&&&&\\
			& \ctrl{1}&\gate{H}    &\meter{$z$} & \cw & \cwbend{-3}  \\
			\lstick{$|\phi\>$}& \targ{} & \qw & \qw & \gate{X}  &\qw &\qw
		\end{quantikz}
		\caption{Local implementation of remote CNOT gate. The wires from top to bottom represent qubits $q$, $q_1$, $q_2$, and $r$ respectively.  Furthermore, $q$ and $q_1$ belong to Alice while $q_2$ and $r$ belong to Bob.}\label{fig-rcnot}
	\end{figure}
	The protocol is depicted as in Fig.~\ref{fig-rcnot} and can be written as a distributed program
	%$\textit{RCNOT}\define \mathit{Alice}\ \|\ \mathit{Bob}$
	%where $\mathit{Alice}\define$
	%\begin{align*}
	%	& q,q_1 \apply \textit{CNOT};\  x_A := \mymeas\ q_1;\ 
	%	stage_A := 0;\\
	%	& \mathbf{do}\ stage_A = 0; c!x_A \ra stage_A := 1 \\
	%	&\ \  \square\  stage_A = 1; d?z_A \ra stage_A := 2; \ \measstm{z_A=1}{q\apply Z}{\sskip} \\
	%	& \mathbf{od}
	%\end{align*}
	%and $\textit{Bob}\define$
	%\begin{align*}
	%	& stage_B := 0;\\
	%	& \mathbf{do}\ stage_B = 0; c?x_B \ra stage_B := 1; \ \measstm{x_B=1}{q_2\apply X}{\sskip};\\
	%	& \hspace{10.5em} q_2, r \apply \textit{CNOT};\ q_2 \apply H; z_B := \mymeas\ q_2\\
	%	&\ \  \square\  stage_B = 1; d!z_B \ra stage_B := 2 \\
	%	& \mathbf{od}
	%\end{align*}
	$\textit{RCNOT}\define \mathit{Alice}\ \|\ \mathit{Bob}$
	where $\mathit{Alice}\define$
	\begin{align*}
		& q,q_1 \apply \textit{CNOT};\  x_A := \mymeas\ q_1;\ 
		stage_A := 0;\\
		& \mathbf{do}\ stage_A = 0; c!x_A \ra stage_A := 1 \\
		&\ \  \square\  stage_A = 1; d?z_A \ra stage_A := 2; \ \measstm{z_A=1}{q\apply Z}{\sskip} \\
		& \mathbf{od}
	\end{align*}
	and $\textit{Bob}\define$
	\begin{align*}
		& q_2,r \apply \textit{CNOT};\  q_2 \apply H;\  z_B := \mymeas\ q_2;\  stage_B := 0;\\
		& \mathbf{do}\ stage_B = 0; c?x_B \ra stage_B := 1; \ \measstm{x_B=1}{r\apply X}{\sskip}\\
		&\ \  \square\  stage_B = 1; d!z_B \ra stage_B := 2 \\
		& \mathbf{od}
	\end{align*}
	The correctness of $\textit{RCNOT}$ is stated as follows: for any $\alpha_{ij}\in \C$ with $\sum_{i,j =0,1}|\alpha_{ij}|^2 = 1$,
	\begin{equation}\label{eq:rcnot}
		\vdash_{\tot}	\ass{\sum_{i,j =0,1}\alpha_{ij} |i,j\>_{q,r}\otimes |\beta\>_{q_1,q_2}}{\mathit{RCNOT}}{\sum_{i,j =0,1}\alpha_{ij} |i,j\oplus i\>_{q,r}}
	\end{equation}
	where $\oplus$ denotes the addition modulo 2.
	Again, the fact that the postcondition does not refer to $q_1$ and $q_2$ means that the post-measurement state of these quantum systems is irrelevant. 
	
	Similar to that of $\textit{Teleport}$, to prove the correctness of $\textit{RCNOT}$ it suffices to show:
	\begin{enumerate}
		\item the cq-assertion \begin{align*}
			\qassertp \define \frac{1}{4}\sum_{i,j =0,1}& \left(\cqs{stage_A = stage_B =0\wedge x_A = i \wedge z_B = j}{|i,j\>_{q_1,q_2} \otimes X^i_rZ^j_{q} |\varphi\>_{q,r}} \right.\\
			+ & \cqs{stage_A =  stage_B =1 \wedge x_A = i \wedge z_B = j}{|i,j\>_{q_1,q_2} \otimes Z^j_{q} |\varphi\>_{q,r}}\\
			+ & \left.\cqs{stage_A =  stage_B =2\wedge x_A = i \wedge z_B = j}{|i,j\>_{q_1,q_2} \otimes  |\varphi\>_{q,r}} \right).
		\end{align*} 
		where $|\varphi\> \define \sum_{k,\ell =0,1}\alpha_{k\ell} |k,\ell\>$, serves as a global invariant for $\mathit{RCNOT}$;
		\item the expression $t\define 2-stage_A$ is a classical ranking function; and
		\item the sequential part of $\mathit{RCNOT}$ establishes $\qassertp$ from the precondition $|\varphi\>_{q,r}\otimes |\beta\>_{q_1,q_2}$.	
	\end{enumerate}
	We omit the details here.
\end{document}

%% file: mymacro.tex
\def\>{\ensuremath{\rangle}}
\def\<{\ensuremath{\langle}}

\def\-{\ensuremath{\textrm{-}}}

\def\apply{\mathrel{*\!\!=}}

\def\change{\ensuremath{\mathit{change}}}

\def\qVar{\ensuremath{\mathit{qv}}}
\def\qv{\ensuremath{\mathit{qv}}}
\def\cVar{\ensuremath{\mathit{cv}}}
\def\QVar{\ensuremath{\mathit{qVar}}}
\def\CVar{\ensuremath{\mathit{cVar}}}
\def\Var{\ensuremath{\mathit{var}}}
\def\Chan{\mathit{chan}}

\def\Exp{\mathit{Exp}}

\def\rassign{:=_{\$}}

\def\h{\ensuremath{\mathcal{H}}}

\def\l{\ensuremath{\mathcal{L}}}

\def\lh{\ensuremath{\mathcal{L(H)}}}
\def\dh{\ensuremath{\mathcal{D(H})}}
\def\dhv{\ensuremath{\d(\h_V)}}
\def\q{\bold Q}

\def\R{\ensuremath{\mathbb{R}}}
\def\m{\ensuremath{\mathcal{M}}}

\def\k{\ensuremath{\mathcal{K}}}

\def\s{\ensuremath{\mathcal{S}}}
\def\t{\ensuremath{\mathcal{T}}}

\def\x{\ensuremath{\mathcal{X}}}

\def\z{\ensuremath{\mathcal{Z}}}

\def\ra{\ensuremath{\rightarrow}}
\def\a{\ensuremath{\mathcal{A}}}
\def\b{\ensuremath{\mathcal{B}}}

\def\e{\ensuremath{\mathcal{E}}}
\def\f{\ensuremath{\mathcal{F}}}
\def\l{\ensuremath{\mathcal{L}}}

\def\N{\mathbb{N}}

\def\Z{\mathbb{Z}}

\def\qiz{\ensuremath{|i\>_q\<0|}}
\def\qzi{\ensuremath{|0\>_q\<i|}}

\def\d{\ensuremath{\mathcal{D}}}
\def\dh{\ensuremath{\mathcal{D(H)}}}
\def\lh{\ensuremath{\mathcal{L(H)}}}
\def\le{\ensuremath{\sqsubseteq}}
\def\ge{\ensuremath{\sqsupseteq}}

\def\next{\mathcal{X}}

\newcommand{\supp}[1]{\ensuremath{\lceil{#1}\rceil}}

\renewcommand{\theenumi}{(\arabic{enumi})}

\newcommand{\tr}{{\rm tr}}
\newcommand{\rto}[1]{\stackrel{#1}\rightarrow}

\newcommand{\ass}[3]{\left\{#1\right\}\ #2\ \left\{#3\right\}}

\newcommand{\conf}[3]{\left\<#1, #2, #3\right\>}
\newcommand{\cqs}[2]{\left\<#1, #2\right\>}

\newcommand {\true} {\ensuremath{{\mathbf{true}}}}
\newcommand {\false} {\ensuremath{{\mathbf{false}}}}
\newcommand {\abort}{\ensuremath{{\mathbf{abort}}}}
\newcommand {\sskip} {\mathbf{skip}}

\newcommand {\ddo} {\ensuremath{\mathbf{do}}}

\newcommand {\mymeas} {\mathbf{meas}}

\newcommand {\tot} {\mathit{tot}}
\newcommand {\pal} {\mathit{par}}
\newcommand {\fail} {\mathbf{fail}}
\newcommand {\iif} {\mathbf{if}}
\newcommand {\fii} {\mathbf{fi}}
\newcommand {\od} {\mathbf{od}}

\newcommand\measstm[3]{\iif\ #1\ra\ #2\ \square\ \neg (#1)\ra #3\ \fii}

\newcommand\alterex{\iif\ B_1\ra S_1\ \square\ldots\square\ B_n\ra S_n\ \fii}

\newcommand\altercom{\iif\ \square_{i=1}^n B_i\ra S_i\ \fii}

\newcommand\repex{\ddo\ B_1\ra S_1\ \square\ldots\square\ B_n\ra S_n\ \od}

\newcommand\repcom{\ddo\ \square_{i=1}^n B_i\ra S_i\ \od}

\newcommand{\define}{\ensuremath{\triangleq}}

\def\tybool{\ensuremath{\mathbf{Boolean}}}
\def\tyint{\ensuremath{\mathbf{Integer}}}
\def\tyqubit{\ensuremath{\mathbf{Qubit}}}
\def\tyqudit{\ensuremath{\mathbf{Qudit}}}

\def\type{\ensuremath{\mathit{type}}}

\def\qstate{\Delta}
\def\qassert{\Theta}
\def\qassertp{\Psi}

\def\cstate{\sigma}
\def\cstates{\Sigma}
\def\cassert{p}
\def\emptydis{\bot}

\def\Exp{\mathrm{Exp}}

\def\qstates{\s_V}
\def\qasserts{\a_V}
\def\qstatesh#1{\mathcal{S}_{#1}}
\def\qassertsh#1{\mathcal{A}_{#1}}

\newcommand\prog{\mathit{Prog}}
\def\ph{\ensuremath{\mathcal{P}(\h)}}
\def\phv{\ensuremath{\mathcal{P}(\h_V)}}

\def\l{\mathcal{L}}
\def\k{\mathcal{K}}

\def\z{\mathbf{0}}

\def\C{\mathbb{C}}

\newcommand{\subs}[2]{{#2}/{#1}}

\def \Rm#1{\mbox{\rm #1}}
\def \lsem      {\raise1pt\hbox{\Rm {[\kern-.12em[}}}
\def \rsem      {\raise1pt\hbox{\Rm {]\kern-.12em]}}}
\def \sem#1{\mbox{\lsem$#1$\rsem}}

%% file: qdp-concur.bbl
\begin{thebibliography}{10}

\bibitem{apt2010verification}
Krzysztof Apt, Frank~S De~Boer, and Ernst-R{\"u}diger Olderog.
\newblock {\em Verification of sequential and concurrent programs}.
\newblock Springer Science \& Business Media, 2010.

\bibitem{Apt86}
Krzysztof~R. Apt.
\newblock Correctness proofs of distributed termination algorithms.
\newblock {\em {ACM} Trans. Program. Lang. Syst.}, 8(3):388--405, 1986.

\bibitem{apt1987two}
Krzysztof~R Apt, Luc Boug{\'e}, and Ph~Clermont.
\newblock {Two normal form theorems for CSP programs}.
\newblock {\em Information Processing Letters}, 26(4):165--171, 1987.

\bibitem{bennett1984quantum}
Charles~H Bennett and Gilles Brassard.
\newblock Quantum cryptography: Public key distribution and coin tossing.
\newblock In {\em Proceedings of the International Conference on Computers,
  Systems and Signal Processing}, 1984.

\bibitem{bennett1993teleporting}
Charles~H Bennett, Gilles Brassard, Claude Cr{\'e}peau, Richard Jozsa, Asher
  Peres, and William~K Wootters.
\newblock {Teleporting an unknown quantum state via dual classical and
  Einstein-Podolsky-Rosen channels}.
\newblock {\em Physical Review Letters}, 70(13):1895, 1993.

\bibitem{bennett1992communication}
Charles~H Bennett and Stephen~J Wiesner.
\newblock {Communication via one-and two-particle operators on
  Einstein-Podolsky-Rosen states}.
\newblock {\em Physical Review Letters}, 69(20):2881, 1992.

\bibitem{briegel1998quantum}
H-J Briegel, Wolfgang D{\"u}r, Juan~I Cirac, and Peter Zoller.
\newblock Quantum repeaters: the role of imperfect local operations in quantum
  communication.
\newblock {\em Physical Review Letters}, 81(26):5932, 1998.

\bibitem{brookes1984theory}
Stephen~D Brookes, Charles~AR Hoare, and Andrew~W Roscoe.
\newblock A theory of communicating sequential processes.
\newblock {\em Journal of the ACM (JACM)}, 31(3):560--599, 1984.

\bibitem{chadha2006reasoning}
Rohit Chadha, Paulo Mateus, and Am{\'\i}lcar Sernadas.
\newblock Reasoning about imperative quantum programs.
\newblock {\em Electronic Notes in Theoretical Computer Science}, 158:19--39,
  2006.

\bibitem{Cuomo+20}
Daniele Cuomo, Marcello Caleffi, and Angela~Sara Cacciapuoti.
\newblock Towards a distributed quantum computing ecosystem.
\newblock {\em IET Quantum Communication}, 1(1):3--8, 2020.

\bibitem{dijkstra1976discipline}
Edsger~Wybe Dijkstra.
\newblock {\em A Discipline of Programming}.
\newblock Prentice-Hall Englewood Cliffs, 1976.

\bibitem{eisert2000optimal}
Jens Eisert, Kurt Jacobs, Polykarpos Papadopoulos, and Martin~B Plenio.
\newblock Optimal local implementation of nonlocal quantum gates.
\newblock {\em Physical Review A}, 62(5):052317, 2000.

\bibitem{feng2007probabilistic}
Yuan Feng, Runyao Duan, Zhengfeng Ji, and Mingsheng Ying.
\newblock Probabilistic bisimulations for quantum processes.
\newblock {\em Information and Computation}, 205(11):1608--1639, 2007.

\bibitem{feng2007proof}
Yuan Feng, Runyao Duan, Zhengfeng Ji, and Mingsheng Ying.
\newblock Proof rules for the correctness of quantum programs.
\newblock {\em Theoretical Computer Science}, 386(1-2):151--166, 2007.

\bibitem{FengDY12}
Yuan Feng, Runyao Duan, and Mingsheng Ying.
\newblock Bisimulation for quantum processes.
\newblock {\em {ACM} Trans. Program. Lang. Syst.}, 34(4):17:1--17:43, 2012.

\bibitem{feng2020quantum}
Yuan Feng and Mingsheng Ying.
\newblock {Quantum Hoare logic with classical variables}.
\newblock {\em ACM Transactions on Quantum Computing, to appear}, 2021.

\bibitem{GayN05}
Simon~J. Gay and Rajagopal Nagarajan.
\newblock Communicating quantum processes.
\newblock In Jens Palsberg and Mart{\'{\i}}n Abadi, editors, {\em Proceedings
  of the 32nd {ACM} {SIGPLAN-SIGACT} Symposium on Principles of Programming
  Languages, {POPL} 2005, Long Beach, California, USA, January 12-14, 2005},
  pages 145--157. {ACM}, 2005.

\bibitem{gottesman1999demonstrating}
Daniel Gottesman and Isaac~L Chuang.
\newblock Demonstrating the viability of universal quantum computation using
  teleportation and single-qubit operations.
\newblock {\em Nature}, 402(6760):390--393, 1999.

\bibitem{hehner2012practical}
Eric~CR Hehner.
\newblock {\em A Practical Theory of Programming}.
\newblock Springer Science \& Business Media, 2012.

\bibitem{hoare1969axiomatic}
Charles Antony~Richard Hoare.
\newblock An axiomatic basis for computer programming.
\newblock {\em Communications of the ACM}, 12(10):576--580, 1969.

\bibitem{hoare1978communicating}
Charles Antony~Richard Hoare.
\newblock Communicating sequential processes.
\newblock {\em Communications of the ACM}, 21(8):666--677, 1978.

\bibitem{ishizaka2008asymptotic}
Satoshi Ishizaka and Tohya Hiroshima.
\newblock Asymptotic teleportation scheme as a universal programmable quantum
  processor.
\newblock {\em Physical Review Letters}, 101(24):240501, 2008.

\bibitem{JorrandL04}
Philippe Jorrand and Marie Lalire.
\newblock Toward a quantum process algebra.
\newblock In Stamatis Vassiliadis, Jean{-}Luc Gaudiot, and Vincenzo Piuri,
  editors, {\em Proceedings of the First Conference on Computing Frontiers,
  2004, Ischia, Italy, April 14-16, 2004}, pages 111--119. {ACM}, 2004.

\bibitem{Kakutani:2009}
Yoshihiko Kakutani.
\newblock {A Logic for Formal Verification of Quantum Programs}.
\newblock {\em Lecture Notes in Computer Science}, pages 79--93, 2009.

\bibitem{Kimble08}
H.~J. Kimble.
\newblock The quantum internet.
\newblock {\em Nature}, 453(7198):1023--1030, 2008.

\bibitem{KozlowskiW19}
Wojciech Kozlowski and Stephanie Wehner.
\newblock Towards large-scale quantum networks.
\newblock In {\em Proceedings of the Sixth Annual {ACM} International
  Conference on Nanoscale Computing and Communication, {NANOCOM} 2019, Dublin,
  Ireland, September 25-27, 2019}, pages 3:1--3:7. {ACM}, 2019.

\bibitem{kraus1983states}
Karl Kraus, Arno B{\"o}hm, John~D Dollard, and WH~Wootters.
\newblock States, effects, and operations: fundamental notions of quantum
  theory.
\newblock {\em Lecture Notes in Physics}, 190, 1983.

\bibitem{nielsen2002quantum}
Michael~A Nielsen and Isaac Chuang.
\newblock {\em Quantum Computation and Quantum Information}.
\newblock Cambridge University Press, 2002.

\bibitem{pirandola2015advances}
Stefano Pirandola, Jens Eisert, Christian Weedbrook, Akira Furusawa, and
  Samuel~L Braunstein.
\newblock Advances in quantum teleportation.
\newblock {\em Nature Photonics}, 9(10):641--652, 2015.

\bibitem{Pompili-3node}
M.~Pompili, S.~L.~N. Hermans, S.~Baier, H.~K.~C. Beukers, P.~C. Humphreys,
  R.~N. Schouten, R.~F.~L. Vermeulen, M.~J. Tiggelman, L.~dos Santos~Martins,
  B.~Dirkse, S.~Wehner, and R.~Hanson.
\newblock Realization of a multinode quantum network of remote solid-state
  qubits.
\newblock {\em Science}, 372(6539):259--264, 2021.

\bibitem{raussendorf2001one}
Robert Raussendorf and Hans~J Briegel.
\newblock A one-way quantum computer.
\newblock {\em Physical Review Letters}, 86(22):5188, 2001.

\bibitem{selinger2004towards}
Peter Selinger.
\newblock Towards a quantum programming language.
\newblock {\em Mathematical Structures in Computer Science}, 14(4):527--586,
  2004.

\bibitem{tafliovich2009programming}
Anya Tafliovich and Eric~CR Hehner.
\newblock Programming with quantum communication.
\newblock {\em Electronic Notes in Theoretical Computer Science},
  253(3):99--118, 2009.

\bibitem{tani2005exact}
Seiichiro Tani, Hirotada Kobayashi, and Keiji Matsumoto.
\newblock Exact quantum algorithms for the leader election problem.
\newblock In {\em Annual Symposium on Theoretical Aspects of Computer Science},
  pages 581--592. Springer, 2005.

\bibitem{unruh2019quantum}
Dominique Unruh.
\newblock Quantum hoare logic with ghost variables.
\newblock In {\em 2019 34th Annual ACM/IEEE Symposium on Logic in Computer
  Science (LICS)}, pages 1--13. IEEE, 2019.

\bibitem{vN55}
John Von~Neumann.
\newblock {\em Mathematical Foundations of Quantum Mechanics}.
\newblock Princeton University Press, Princeton, NJ, 1955.

\bibitem{Wehner18internet}
Stephanie Wehner, David Elkouss, and Ronald Hanson.
\newblock Quantum internet: A vision for the road ahead.
\newblock {\em Science}, 362(6412), 2018.

\bibitem{wootters1982single}
William~K Wootters and Wojciech~H Zurek.
\newblock A single quantum cannot be cloned.
\newblock {\em Nature}, 299(5886):802--803, 1982.

\bibitem{ying2012floyd}
Mingsheng Ying.
\newblock {Floyd--Hoare logic for quantum programs}.
\newblock {\em ACM Transactions on Programming Languages and Systems (TOPLAS)},
  33(6):1--49, 2012.

\bibitem{ying2016foundations}
Mingsheng Ying.
\newblock {\em Foundations of Quantum Programming}.
\newblock Morgan Kaufmann, 2016.

\bibitem{ying2019toward}
Mingsheng Ying.
\newblock Toward automatic verification of quantum programs.
\newblock {\em Formal Aspects of Computing}, 31(1):3--25, 2019.

\bibitem{ying2009algebra}
Mingsheng Ying, Yuan Feng, Runyao Duan, and Zhengfeng Ji.
\newblock An algebra of quantum processes.
\newblock {\em ACM Transactions on Computational Logic (TOCL)}, 10(3):1--36,
  2009.

\bibitem{zoebel1988normalform}
Dieter Zoebel.
\newblock {Normalform-Transformationen f{\"u}r CSP-Programme}.
\newblock {\em Informatik (Berlin, West)}, 3(2):64--76, 1988.

\end{thebibliography}
